\documentclass[preprint,12pt,authoryear]{elsarticle}

\usepackage{amssymb}
\usepackage{amsmath}
\usepackage{graphicx} % Required for inserting images
\usepackage{amsfonts}
\usepackage{bbm}
\usepackage{xcolor}
\usepackage{amsthm}
\usepackage{booktabs}
\usepackage[hidelinks,colorlinks=true,linkcolor=blue,citecolor=blue]{hyperref}
\usepackage{natbib}

\newtheorem{theorem}{Theorem}[section]
\newtheorem{lemma}[theorem]{Lemma}

\journal{Nuclear Physics B}

\newcommand{\diff}{\,\mathrm{d}}
\newcommand{\var}{\mathbb{V}\text{ar}}
\DeclareMathOperator*{\argmin}{\arg\!\min}
\DeclareMathOperator*{\argmax}{\arg\!\max}

\newcommand{\RomanNumeralCaps}[1]
    {\text{\MakeUppercase{\romannumeral #1}}}

\allowdisplaybreaks

\begin{document}

\begin{frontmatter}



\title{Time-Averaged Drift Approximations are Inconsistent for Inference in Drift Diffusion Models} %% Article title


\author[label1]{Sicheng Liu\corref{cor1}} %% Author name
\ead{sicheng_liu@brown.edu}
\author[label2]{Alexander Fengler}
\ead{alexander_fengler@brown.edu}
\author[label2,label3]{Michael J. Frank}
\ead{michael_frank@brown.edu}
\author[label1]{Matthew T. Harrison}
\ead{matthew_harrison@brown.edu}

\cortext[cor1]{Corresponding author.}
\affiliation[label1]{organization={Division of Applied Mathematics, Brown University},%Department and Organization
            addressline={182 George St}, 
            city={Providence},
            postcode={02912}, 
            state={RI},
            country={USA}}
\affiliation[label2]{organization={Department of Cognitive and Psychological Sciences, Brown University},%Department and Organization
            addressline={190 Thayer St}, 
            city={Providence},
            postcode={02912}, 
            state={RI},
            country={USA}}
\affiliation[label3]{organization={Carney Institute for Brain Science, Brown University},%Department and Organization
            addressline={164 Angell St}, 
            city={Providence},
            postcode={02912}, 
            state={RI},
            country={USA}}

%% Abstract
\begin{abstract}
 Drift diffusion models (DDMs) have found widespread use in computational neuroscience, cognitive science, mathematical psychology as well as other fields. They model evidence accumulation in simple decision tasks as a stochastic process drifting towards decision barriers. In models where the drift is both time-varying within a trial and variable across trials, the high computational cost for accurate likelihood evaluation has often led to the use of a computationally convenient surrogate for parameter inference, the time-averaged drift approximation (TADA). In each trial, TADA assumes that the time-varying drift rate can be replaced by its temporal average throughout the trial. This approach enables fast parameter inference using analytical likelihood formulas for DDMs with constant drift. In this work, we show that such an estimator is inconsistent: it does not converge to the true drift, posing a risk of biasing scientific conclusions when parameter estimates are obtained by TADA and similar approximations. We provide an elementary proof of this inconsistency in what is perhaps the simplest possible setting: a Brownian motion with piecewise constant drift hitting a one-sided upper boundary. Furthermore, numerical examples based on an attentional DDM (aDDM) show that using TADA leads to systematic misestimation of attentional effects in decision making and can lead to false conclusions in scientific hypothesis testing.

% we show that this example exhibits severe inconsistency: the estimator will always converge to a value strictly greater than the true drift parameter. 
\end{abstract}


% \begin{highlights}
% \item Drift diffusion models are commonly used to model simple decision-making scenarios.
% \item A common computational shortcut replaces a time-varying drift with a constant drift.
% \item Counterexamples show inconsistent statistical estimates by proof and simulation.
% \item Systematic bias can occur when used with attentional drift diffusion models.
% \item Caution is recommended when using the shortcut to interpret experimental data.
% \end{highlights}
% \begin{highlights}
% \item We study a common unjustified computation shortcut in diffusion model fitting.
% \item It replaces the time-varying drift with a constant ``effective'' drift. 
% \item It is not a well specified model for first passage time data. 
% \item Counterexamples show inconsistent maximum likelihood estimates by proof and simulation.
% %\item Simulation quantifies the bias in estimating attention effects.
% \item We caution that attention effects can be misestimated and conclusions can be misleading.
% \end{highlights}

%% Keywords

\end{frontmatter}



\section{Introduction}
%\subsection{Drift Diffusion Model}
% \textcolor{blue}{[Description of DDM (copy paste from NESS2024 poster), need to be rewritten later; reduce the length or delete if submitted to a neuroscience journal]}

\textbf{Drift diffusion models (DDMs)} are an important class of models that describe the latent psychological processes underlying simple decision-making scenarios such as two-alternative forced choice experiments \citep{ratcliff2008, ratcliff1978, ratcliff2016}. In such experiments, a subject must choose between two alternatives, A and B, and the experimenter records both the choice and the time taken to make it. DDMs are used to model the joint distribution over response times and choices and are widely applied in empirical research \citep{Peters2019, Clithero2016Response, Myers2022A, Pedersen2017The, milosavljevic2010drift} to extract their interpretable parameters, which can in turn be linked to distinct cognitive and neural mechanisms. 

The standard DDM is described by the stochastic differential equation
\begin{equation}\label{eqn: ddm}
\diff X(t)=\mu \diff t+\sigma \diff W(t)
\end{equation}
starting at $X(0)=x_0$, where $W(t)$ is a 1D Brownian motion. Decision dynamics are captured by $X(t)$ before the stopping time $\tau$, where 
\begin{equation*}
\tau\triangleq \inf \{t>0: X(t)\notin (\ell, u)\}    
\end{equation*}
is the first hitting time of $X(t)$ to the upper boundary $u$ or the lower boundary $\ell$, indicating the decision time. The boundary reached at time $\tau$ reflects the chosen outcome, which we denote as $C\in\{u, \ell\}$. The parameter vector $\boldsymbol{\theta}$ may include $\mu, \sigma, x_0, u, \ell$ or any further parameters that define them. The likelihood function is the \textbf{first passage time density (FPTD)} at the boundary $C$ evaluated at $\tau$ when viewed as a function of the parameter vector $\boldsymbol{\theta}$. For the standard DDM, the analytic solution for the density and distribution of $\tau$ can be mathematically derived in closed form \citep{navarro2009fast, gondan2014even, cox2017theory, blurton2012fast}, enabling tractable likelihood-based inference for its parameters. As a result, this model has found widespread applications in cognitive neuroscience and mathematical psychology \citep{ratcliff1978, ratcliff2008, cavanagh2014eye, frank2015, doi2020, yartsev2018, fengler2022beyond}.

The standard DDM can be extended in a variety of ways to better capture experimental data and to test hypotheses about the biological mechanisms underlying decision-making. Examples include models with time-varying decision boundaries $u(t)$ and $\ell(t)$ instead of constant thresholds $u$ and $\ell$ \citep{cisek2009decisions, Drugowitsch2012, malhotra2018, palestro2018, ging-jehli_basal_ganglia_2025}, models with more complex stochastic dynamics than Brownian motion \citep{wieschen2020jumping, rasanan2024there}, and models with trial-by-trial drift rates that covary with exogenous covariates, where each trial $i$ has its own drift $\mu_i$ or $\mu_i(t)$, rather than sharing a common drift $\mu$ or $\mu(t)$ across trials
 \citep{Turner2015Informing, frank2015, wiecki2013hddm, fengler2022beyond}. Many of these extensions typically do not admit closed-form expressions for their FPTDs, but  numerical methods have been developed to compute the FPTDs of DDMs to facilitate their application, including approaches based on solving Kolmogorov partial differential equations (PDEs) \citep{shinn2020flexible, diederich2003simple, voss2007fast, voss2008fast, rasanan2023numerical} and integral equations (IEs) \citep{smith2000stochastic, smith2022modeling, paninski2008integral, peskir2002integral}. For a comprehensive review of these methods, we refer the readers to \cite{richter2023diffusion}.

When the parameters describing a DDM are both time-varying and trial-varying, for example, when the drift rate is modulated by covariates that vary on a moment-by-moment basis in time, the computational cost of computing different FPTDs on each trial can become too large for practical use. A common approach in this situation is to approximate a time-varying-parameter model with a constant-parameter model that permits fast FPTD evaluation. An open question in the literature is the degree to which these commonly-used approximations lead to valid conclusions about the underlying time-varying models of scientific interest. 

In this paper, we demonstrate both mathematically and empirically that constant-parameter approximations can induce substantial statistical bias in parameter inference and, in particular, lead to erroneous conclusions in hypothesis testing, raising serious concerns about their use. We focus on a class of DDMs called attentional DDMs (aDDMs) and a class of constant-parameter approximations that we call time-averaged drift approximations (TADA). As described in Section \ref{sec: addm}, aDDMs have time-varying and trial-varying drift rates, and TADA for the aDDM replaces the time-varying drift on each trial with its temporal average over the trial. We focus on these example cases for the following reasons. (i) They are of current practical importance. (ii) Simple versions facilitate exact mathematical analysis of large-sample statistical behavior. (iii) Most existing uses of constant-parameter approximations arise specifically in aDDMs via approaches related to TADA. (iv) In parallel work \citep{liu2025efficient} we developed \textsc{efpt}, an efficient method called efpt for computing FPTDs in aDDMs, allowing us to numerically compare true and approximate methods for this class of models with ease.

\section{Attentional Drift Diffusion Models}\label{sec: addm}
The \textbf{attentional drift diffusion model (aDDM)} incorporates visual attention as a key factor underlying the decision process behind simple choice behavior \citep{krajbich2010visual}. It provides a canonical example of a DDM with trial-varying and time-varying drift: within each trial, the drift rate $\mu(t)$ is modulated by instantaneous gaze, taking different values depending on whether the participant is fixating option A or option B at time $t$. Specifically, if we let $V(t)$ be the process taking values in $\{\mathrm{A}, \mathrm{B}\}$ representing the participant's visual fixation at time $t$, and let $r_A$ and $r_B$ denote the participant's numerical ratings of their preferences for options A and B, respectively. The aDDM corresponds to the following SDE (conditioned on $V, r_A, r_B$) \citep{krajbich2010visual},
\begin{equation}\label{eqn: addm}
	\begin{aligned}
		&\diff X(t)=A(V(t), r_A, r_B) \diff t+\sigma \diff W(t)
	\end{aligned}
\end{equation}
where
\begin{equation*}
A(v, r_A, r_B)=\begin{cases}
			\kappa(r_A-\eta r_B) & \text{if }v=\mathrm{A} ,\\
			\kappa(\eta r_A-r_B) & \text{if }v=\mathrm{B} ,\\
		\end{cases}
\end{equation*}
$\eta\in(0, 1)$ is a parameter that reflects the discounting of the non-fixated option, and $\kappa>0$ parametrizes the overall magnitude of the drift. 

In a typical experimental setting, a participant's fixation trajectory $V(t)$ is captured using an eye-tracker, while the participant's preference ratings $r_A$ and $r_B$ are collected prior to the choice experiment; together, these are treated as known covariates. Note that $\mu(t)=A(V(t), r_A, r_B)$ is a piecewise constant function that is generally different on each experimental trial, both because a participant's visual fixation process will generally be different on each trial and because the options presented to the participant can change across trials. As in Eq.~\eqref{eqn: addm} we often suppress the dependence on trials in our notation.

The dependence on within-trial covariates, here specifically the fixation times, poses significant computational challenges. The analytical results available for the standard DDM no longer apply, as the drift term in the aDDM varies over time. Moreover, because fixation trajectories are unique for each trial, computations cannot be shared across trials, necessitating likelihood evaluation on a trial-by-trial basis. As a result, existing numerical methods, such as the Kolmogorov equation methods and integral equation methods, become impractical for the large datasets typically encountered in experimental studies. Parameter inference in the aDDM can take days or even weeks for sizable datasets, severely limiting its applicability despite its strong scientific interest. This difficulty is further compounded by the discontinuous nature of fixation-dependent drift, which can complicate the use of standard discretization for PDE/IE based methods.

\section{Time-averaged Drift Approximation}\label{sec: glam}

% \textcolor{blue}{[This section needs more concrete cites]}

Due to these aforementioned computational challenges, researchers have resorted to computationally efficient approximations of the aDDM for likelihood computation to enable statistical inference. One common approach is to adopt what we call the \textbf{time-averaged drift approximation (TADA)}. Broadly speaking, this method aims to simplify the computation by summarizing the fixation process $V(t)$ through the proportion of total gaze time allocated to each option to construct a constant ``effective'' drift, rather than modeling the full sequence of fixation trajectories. Then, the analytical likelihood formula of a standard DDM with that ``effective'' drift is used to compute the trial likelihood. This approach simplifies away the computational difficulty associated with within-trial covariates, focusing instead on the overall influence of gaze for each trial, and allows for fast statistical computations in practice. 

It is worth noting that implementation details may vary across studies employing TADA, although the basic idea is shared. For instance, the piecewise constant averaging method \citep{lombardi_hare_2021} conjectures that for each sample path, conditioned on its first passage time $\tau$, a DDM with drift function $\mu(t)$ can be viewed as a standard DDM whose drift coefficient is given by 
\begin{equation}\label{eqn: equi_drift}
\begin{aligned}
\frac{1}{\tau}\int_0^\tau \mu(t)\diff t,
\end{aligned}
\end{equation}
and since the aDDM has a piecewise constant drift function $\mu(t)=A(V(t), r_A, r_B)$, this ``effective'' drift ends up being
\begin{equation}
\begin{aligned}
&\frac{1}{\tau}\int_0^\tau A(V(t), r_A, r_B)\diff t=\\&\quad\frac{\text{time fixated at A}}{\tau}A(\text{A}, r_A, r_B)+\frac{\text{time fixated at B}}{\tau}A(\text{B}, r_A, r_B),
\end{aligned}
\end{equation}
i.e., the temporal average drift in the aDDM weighted by gaze times. The gaze-weighted linear accumulator model (GLAM) \citep{thomas2019gaze, molter2019glambox} applies an additional sigmoid transformation to map the gaze-weighted option values to the “effective” drift rate. GLAM also generalizes the DDM to multiple choice options.
Nevertheless, the essential assumption underlying all such variants is that only the gaze proportions matter for inference, a simplification that substantially improves computational tractability. 

% The logic extends to other gaze-weighted models in the literature, such as e.g. logistic regression models of choice proportion \citep{smith2019estimating}. %Throughout this paper, we use the term TADA to refer broadly to all models that adopt this conceptual framework.

However, despite its widespread usage, the suitability of using TADA to approximate the aDDM \eqref{eqn: addm} has not been quantitatively investigated.   To begin with, regardless of how the ``effective'' drift is defined, TADA is not a well-specified generative model for first passage times since its ``effective'' drift (e.g., in equation \eqref{eqn: equi_drift}) depends on the first passage time $\tau$ itself. Consequently, the so-called TADA ``likelihood'' is not a valid likelihood function (see also Figure \ref{fig: llhd_comparison}) and its corresponding estimators can only be defined in a formal sense. Nevertheless, it has been argued that the aDDM and its piecewise constant averaging variant are equivalent \citep{lombardi_hare_2021}. More commonly, researchers have used TADA as a computational shortcut to fit and test the aDDM \citep{lupkin2023monkeys, sepulveda2020visual, ramirez2022optimal, yang2023dynamic, yang2024attention, ting2024unraveling, zilker2024attentional, thomas2019gaze, thomas2021uncovering, cavanagh2014eye}, implicitly suggesting that, although their likelihoods may differ for individual sample paths, TADA provides a reasonable approximation to the aDDM at the population level of all trial-level covariates. Under this assumption, the maximum likelihood or Bayesian estimators derived from TADA could remain consistent, making it a useful surrogate model for statistical inference and hypothesis testing. 

In this work we show that both statements are technically invalid: TADA and the aDDM have different likelihood functions (unsurprisingly), and the TADA ``maximum likelihood'' estimator is inconsistent, as demonstrated by a concrete counterexample in Section \ref{sec: example}. Numerical experiments further demonstrate this inconsistency and show how it can, in some circumstances, lead to false conclusions in scientific hypothesis testing. Although our analysis does not necessarily invalidate all prior qualitative conclusions based on TADA, it raises serious concerns about its use.

\section{Mathematical Analysis}

In this section we formally prove that a statistical estimator based on a TADA can be statistically inconsistent, meaning that the estimator converges to the wrong value as the number of trials increases. We illustrate inconsistency in a simple example that permits exact mathematical analysis.  
%For the counterexample, we consider the case where only one constant boundary is present, as it is sufficient to demonstrate the inconsistency and simplifies the analytical formulation.

\subsection{Preliminaries}
Here we first summarize the analytical formulas used in the next subsection:
\begin{theorem}\label{thm: 2.1}
Let $X(t)=x_0+\mu t+\sigma W(t)$, and for $b\neq x_0$, let $\tau_b=\inf\{t\ge0:X(t)=b\}$ be the first passage time of $X(t)$ to $b$, then 
\begin{enumerate}
\item For $x$ satisfying $(x-b)(x_0-b)>0$ ($x$ and $x_0$ are on the same side of $b$), the non-passage density (NPD), which is defined as the sub-probability density of $X(t)$ when the process has not yet hit $b$, is given by
\begin{equation}\label{eqn: thm211}
\begin{aligned}
&\mathbb{P}^{x_0}(X(t)\in  \diff  x, \tau_b>t)\\
=&\frac{1}{\sqrt{2 \pi \sigma^2 t}} e^{\frac{\mu}{\sigma^2} (x-x_0)-\frac{\mu^2}{2\sigma^2}t}\Big(e^{-\frac{\left(x-x_0\right)^2}{2 \sigma^2 t}}-e^{-\frac{\left(x+x_0-2 b\right)^2}{2 \sigma^2 t}}\Big) \diff  x
\end{aligned}
\end{equation}
where $\mathbb{P}^{x_0}$ denotes the probability measure under which the process starts at $x_0$.

    \item For $t>0$, the first passage time density (FPTD), which is defined as the density of $\tau_b$, is given by
\begin{equation}\label{eqn: thm212}
\begin{aligned}
\mathbb{P}^{x_0}(\tau_b\in  \diff  t)=&\frac{|b-x_0|}{\sqrt{2\pi \sigma^2 t^3}}e^{-\frac{(b-x_0-\mu t)^2}{2\sigma^2t}} \diff  t\\
\end{aligned}
\end{equation}
When $\mu(b-x_0)>0$, $\tau_b$ is an inverse Gaussian random variable\footnote{The inverse Gaussian distribution $\mathrm{IG}(m, \lambda)$ is a family of continuous probability distributions, where $m>0$ is the mean and $\lambda>0$ is the shape parameter. It is supported on $x\in (0, \infty)$, with probability density function and cumulative distribution function given by
\begin{equation*}
\begin{aligned}
f(x ; m, \lambda)&=\sqrt{\tfrac{\lambda}{2 \pi x^3}} e^{-\frac{\lambda(x-m)^2}{2 m^2 x}}\\
F(x ; m, \lambda)&=1-\tfrac{1}{2} \operatorname{Erfc}\Big(\sqrt{\tfrac{\lambda
}{2x}}\Big(\tfrac{x}{m}-1\Big)\Big)+\tfrac{1}{2} e^{\frac{2\lambda}{m}}\operatorname{Erfc}\Big(\sqrt{\tfrac{\lambda
}{2x}}\Big(\tfrac{x}{m}+1\Big)\Big)
\end{aligned}
\end{equation*}
respectively, where $\operatorname{Erfc}$ is the complementary error function defined by $\operatorname{Erfc}(x)=\frac{2}{\sqrt{\pi}} \int_x^{\infty} e^{-z^2} \diff z$. Its $m=\infty$ limit is the L\'{e}vy distribution, corresponding to the special case $\mu=0$ in Theorem \ref{thm: 2.1}, which directly follows from the reflection principle of Brownian motion.}\\$\operatorname{IG}\left(\frac{b-x_0}{\mu},\left(\frac{b-x_0}{\sigma}\right)^2\right)$;
When $\mu(b-x_0)<0$, Eq.\eqref{eqn: thm212} is a subprobability density as $\mathbb{P}(\tau_b<\infty)=e^{\frac{2(b-x_0) \mu}{\sigma^2}}<1$.
\end{enumerate}
\end{theorem}
Formula \eqref{eqn: thm212} can be derived via various approaches, for example see page 96 of \cite{karatzas1991brownian} or Chapter 5.4 of \cite{cox2017theory}. The classical approach involves using the optional sampling theorem to study the Laplace transform of $\tau_b$ by constructing suitable exponential martingales. For completeness we include a self-consistent proof of both \eqref{eqn: thm211} and \eqref{eqn: thm212} in \ref{sec: proof_prelim}.

\subsection{An Inconsistent Example}\label{sec: example}
Next, we give a simple example to show our two core results:

\begin{enumerate}
    \item Conditional on the first passage time $\tau$, TADA and the aDDM produce different first passage time densities (FPTDs), so TADA cannot be viewed as equivalent to the aDDM per sample path. In fact, the TADA ``FPTD'' function does not even integrate to 1 with respect to $\tau$.
    \item In this example, as the sample size $n\rightarrow\infty$, the TADA estimator always converges to a value strictly greater than the true drift parameter. Hence, it is statistically inconsistent.
\end{enumerate}

We consider a Brownian motion with piecewise constant drift given by
\begin{equation}\label{eqn: model1}
\diff X(t)=\mu\mathbbm{1}_{[0, T]}(t) \diff t+\diff W(t)
\end{equation}
starting at $X(0)=0$ and a constant upper boundary $u(t)=b$, where $b>0$ and $T>0$ are known constants, $\mu \in \mathbb{R}$ is the drift parameter to be inferred, and $\mathbbm{1}_{[0, T]}(t)$ denotes the indicator function that takes value $1$ if $t \in[0, T]$ and $0$ otherwise. Let $\tau=\inf\{t>0: X(t)=b\}$ be the first passage time of $X(t)$ to level $b$. % We also require that $\mu>0$. % so that $X(t)$ will hit $b$ with probability 1. 
This can be viewed either as an aDDM with only one boundary or as an attentional race model with only one accumulator; in both cases, there is one fixed saccade.

We also define 
\begin{equation}
\alpha(\mu, t,T)=\begin{cases}
\mu&t\le T\\
\frac{T}{t}\mu &t>T
\end{cases}
\end{equation}
to represent the average of the drift term until time $t$. This is the ``effective'' drift \eqref{eqn: equi_drift} constructed by the piecewise constant averaging method \citep{lombardi_hare_2021}, which we adopt as the representative form of TADA in our subsequent analysis.

\subsubsection{Non-equivalence of FPTDs}\label{sec: non-equivalence}
    When $\tau\le T$, $\tau$ is distributed according to Eq.\eqref{eqn: thm212}; when $\tau>T$, we view this process as a DDM with a random initial position $x_0$, where its density is specified by the NPD in Eq.\eqref{eqn: thm211}. To obtain the distribution of $\tau$ in this case, we integrate Eq.\eqref{eqn: thm212} over the NPD. To be more specific, when $\tau>T$, the FPTD $f_{\text{aDDM}}(\tau;\mu)$ is given by $\int_{-\infty}^b \frac{1}{\sqrt{2\pi T}}e^{\mu x-\frac{\mu^2}{2}T} \Big(e^{-\frac{x^2}{2T}}-e^{-\frac{(x-2b)^2}{2T}}\Big)\frac{b-x}{\sqrt{2\pi (\tau-T)^3}} e^{-\frac{(b-x)^2}{2(\tau-T)}} \diff x$, which simplifies to
    % As a result, the likelihood of $\mu$ given datum $\tau$ under model \eqref{eqn: model1} is the FPTD at $\tau$ as a function of the parameter $\mu$, given by
\begin{equation}\label{eqn: addmllhd}
\begin{aligned}
&f_{\text{aDDM}}(\tau;\mu)\\
=&\begin{cases}
\frac{b}{\sqrt{2\pi \tau^3}}e^{-\frac{(b-\mu\tau)^2}{2\tau}} & \tau \le T\\ 
\begin{array}{r}
\frac{1}{2 \sqrt{2 \pi\tau^3}}e^{-\frac{(T\mu-b)^2}{2 \tau}} \left[ e^{\frac{2 b\mu(\tau-T)}{\tau}}(T\mu+b) 
\operatorname{Erfc}\left((T\mu+b) \sqrt{\frac{\tau-T}{2T\tau}}\right) \right.  \\
\left. - (T\mu-b) \operatorname{Erfc}\left((T\mu-b) \sqrt{\frac{\tau-T}{2T\tau}}\right) \right] 
\end{array} & \tau > T
\end{cases}
\end{aligned}    
\end{equation}
where $\operatorname{Erfc}(x)=\frac{2}{\sqrt{\pi}}\int_x^\infty e^{-z^2}\diff z$ is the complementary error function\footnote{We use $\operatorname{Erfc}$ throughout this paper for notational consistency. However, it can be equivalently expressed in terms of the error function $\operatorname{Erf}$ or the normal CDF $\Phi$.} (see details of this calculation in \ref{sec: deriv1}). In contrast, the TADA ``FPTD'' is
\begin{equation}\label{eqn: glamllhd}
\begin{aligned}
f_{\text{TADA}}(\tau;\mu)&=\frac{b}{\sqrt{2\pi \tau^3}}e^{-\frac{(b-\alpha(\mu, \tau, T) \tau)^2}{2\tau}}=\begin{cases}
\frac{b}{\sqrt{2\pi \tau^3}}e^{-\frac{(b-\mu \tau)^2}{2\tau}}& \tau \le T\\ 
\frac{b}{\sqrt{2\pi \tau^3}}e^{-\frac{(b-T\mu)^2}{2\tau}}& \tau > T\\ 
\end{cases}
\end{aligned}
\end{equation}
When $\tau\le T$, Eq.\eqref{eqn: addmllhd} and Eq.\eqref{eqn: glamllhd} are the same, which corresponds to the case where the participant never switches their visual fixation. To compare Eq.\eqref{eqn: addmllhd} and Eq.\eqref{eqn: glamllhd} when $\tau>T$, we study their integrals w.r.t. $\tau$ from $T$ to $\infty$.
The true probability of hitting after $T$ is given by
\begin{equation}\label{eqn: survival_prob1}
\begin{aligned}
\mathbb{P}(\tau> T)=&1-\int_0^T \frac{b}{\sqrt{2\pi t^3}}e^{-\frac{(b-\mu t)^2}{2t}} \diff t\\
=&\frac{1}{2}\operatorname{Erfc}\Big(\frac{T \mu-b}{\sqrt{2T}}\Big)-\frac{1}{2}e^{2 b \mu} \operatorname{Erfc}\Big(\frac{T \mu+b}{\sqrt{2T} }\Big)
\end{aligned}
\end{equation}
as Eq.\eqref{eqn: addmllhd} is a valid probability density when viewed as a function of $\tau$. However, integrating Eq.\eqref{eqn: glamllhd} w.r.t. $\tau$ over $(T,\infty)$ gives
\begin{equation}\label{eqn: survival_prob2}
\int_T^\infty \frac{b}{\sqrt{2\pi t^3}}e^{-\frac{(b-T\mu )^2}{2t}} \diff t\\
=\begin{cases}
\frac{b}{T\mu-b}\Big(1-\operatorname{Erfc}\Big(\frac{T\mu-b}{\sqrt{2T}}\Big)\Big)& T\mu\neq b\\
    \sqrt{\frac{2  }{\pi T}}b & T\mu=b
\end{cases}
\end{equation}
We compare Eq.\eqref{eqn: survival_prob1} and Eq.\eqref{eqn: survival_prob2} when $T\mu=b$: \eqref{eqn: survival_prob1} reduces to $\frac{1}{2}\Big(1-e^{\frac{2 b^2}{T}} \operatorname{Erfc}\Big(\sqrt{\frac{2}{T}}b\Big)\Big)$ which is always $\le \frac{1}{2}$; In contrast, Eq.\eqref{eqn: survival_prob2} gives $\sqrt{\frac{2}{\pi T}}b$, which can be any positive number as $b$ and $T$ vary within $(0, \infty)$. As a result, the aDDM FPTD Eq.\eqref{eqn: addmllhd} and the TADA ``FPTD'' Eq.\eqref{eqn: glamllhd} are different for $\tau>T$ (see Fig.\ref{fig: llhd_comparison} for illustration). Hence, the aDDM and TADA are not equivalent, except in the trivial case with no saccades.

\begin{figure}[htbp]
  \centering
  \includegraphics[width=0.75\linewidth]{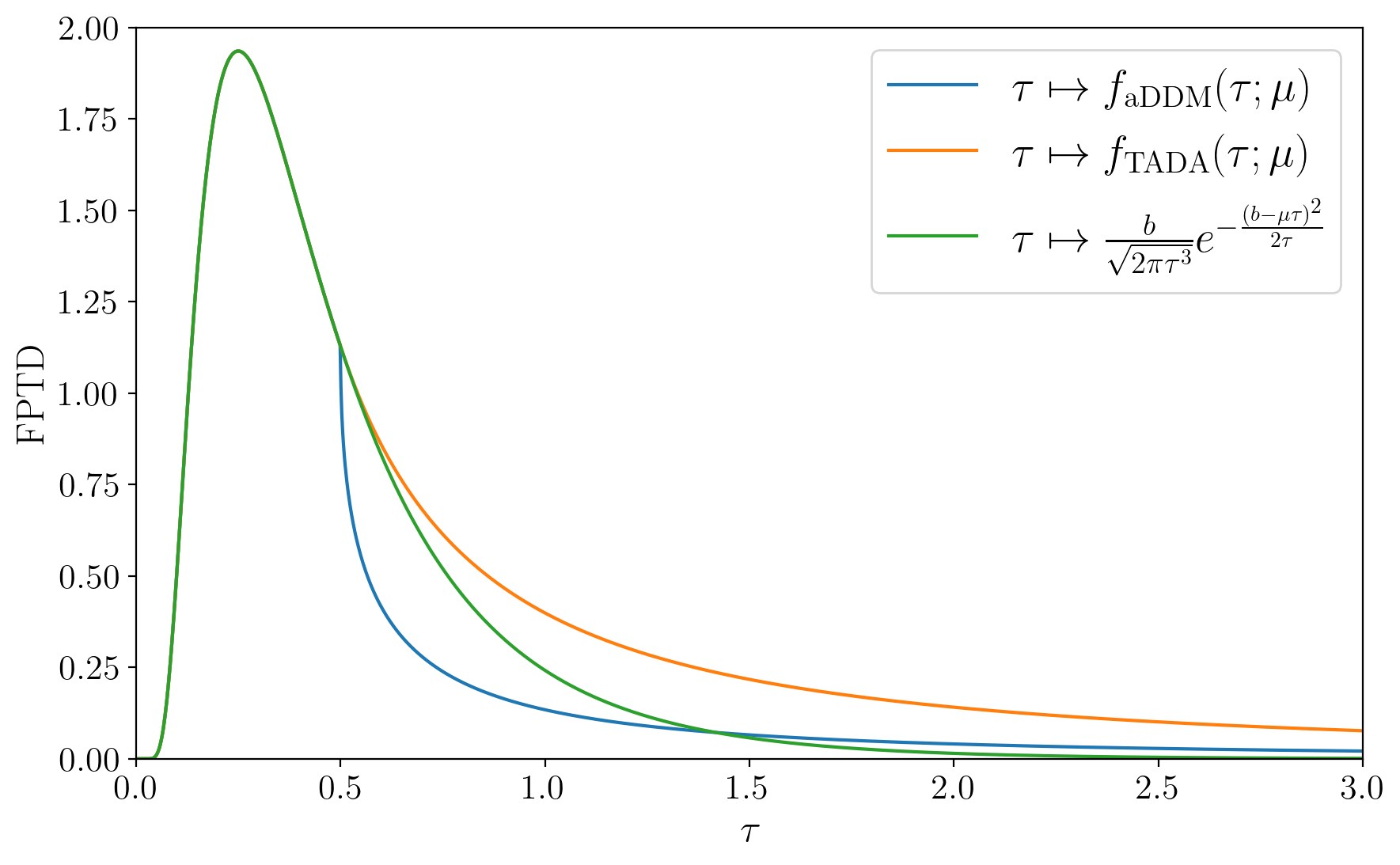}
  \caption{Plots of $f_{\text{aDDM}}$\eqref{eqn: addmllhd} (blue curve), $f_{\text{TADA}}$\eqref{eqn: glamllhd} (orange curve) as functions of $\tau$ when $\mu=2, b=1, T=0.5$. The additional green curve is the FPTD when the drift is constant $\mu$ instead of $\mu\mathbbm{1}_{[0,T]}(t)$. All three functions are the same when $\tau\le T$, but become different when $\tau>T$. The area under $\tau\mapsto f_{\text{aDDM}}(\tau;\mu)$ and $\tau\mapsto\frac{b}{\sqrt{2\pi\tau^3}} e^{-\frac{(b-\mu\tau)^2}{2\tau}}$ are each 1 while that of $\tau\mapsto f_{\text{TADA}}(\tau;\mu)$ is not 1. This discrepancy arises because TADA is not a well-specified model, and $f_{\text{TADA}}(\tau;\mu)$ is not a proper probability density function.}
  \label{fig: llhd_comparison}
\end{figure}

\subsubsection{Inconsistency of the TADA estimator}

In this section, we study the inference of parameter $\mu$ using the dataset $(\tau_i)_{i=1}^n$, where $\tau_1,\cdots,\tau_n$ are independent and identically distributed as $\tau$. 

We can write the maximum likelihood estimator (MLE)
\begin{equation}\label{eqn: addm_mle}
\widehat{\mu}_{n}^{\text{ML}}=\argmax_{\mu} \prod_{i=1}^n f_{\text{aDDM}}(\tau_i;\mu)
\end{equation}
which maximizes the likelihood in Eq.\eqref{eqn: addmllhd} on dataset $(\tau_i)_{i=1}^n$, and the corresponding TADA ``maximum likelihood'' estimator
\begin{equation}\label{eqn: glam_mle}
\widehat{\mu}_{n}^{\text{TADA}}=\argmax_{\mu} \prod_{i=1}^n f_{\text{TADA}}(\tau_i; \mu)
\end{equation}
The consistency of $\widehat{\mu}_n^{\text{ML}}$ follows from known results about the consistency of MLEs (see \ref{sec: regularities}). 

Here we prove that $\widehat{\mu}_n^{\text{TADA}}$ is not a consistent estimator of $\mu$. We write the TADA negative log-likelihood
\begin{equation*}
\begin{aligned}
&\ell_n^{\text{TADA}}(\mu\mid\tau_1, \cdots, \tau_n)\\
&\quad\triangleq -\sum_{i=1}^n \log f_{\text{TADA}}(\tau_i;\mu)\\
&\quad=-n\log(b)+\frac{n}{2}\log(2\pi) +\sum_{i=1}^n\Big(\frac{3}{2}\log\tau_i+\frac{1}{2\tau_i}(b-\alpha(\mu, \tau_i, T) \tau_i)^2\Big)
\end{aligned}
\end{equation*}
and $\widehat{\mu}_{n}^{\text{TADA}}$ \eqref{eqn: glam_mle} is given by
\begin{equation}
\begin{aligned}
\widehat{\mu}_{n}^{\text{TADA}}&=\argmin_{\mu} \ell_n^{\text{TADA}}(\mu\mid\tau_1, \cdots, \tau_n).
\end{aligned}
\end{equation}
Setting $\frac{\partial \ell_n^{\text{TADA}}}{\partial \mu}$ to $0$, we get
\begin{equation}\label{eqn: 7}
\begin{aligned}
0=&\frac{\partial \ell_n^{\text{TADA}}}{\partial \mu}(\widehat{\mu}_{n}^{\text{TADA}}\mid\tau_1, \cdots, \tau_n)\\
=&\sum_{i=1}^n \alpha_\mu(\widehat{\mu}_{n}^{\text{TADA}}, \tau_i, T)(\alpha(\widehat{\mu}_{n}^{\text{TADA}}, \tau_i, T)\tau_i-b),
\end{aligned}
\end{equation}
where
\begin{equation*}
\alpha_\mu(\mu, \tau, T)=\frac{\partial}{\partial\mu}\alpha(\mu, \tau, T)=\begin{cases}
1& t\le T\\
\frac{T}{t} & t>T
\end{cases}.
\end{equation*}
One can verify that the second derivative is $\sum_{i=1}^n (\alpha_\mu(\widehat{\mu}_{n}^{\text{TADA}}, \tau_i, T))^2\tau_i>0$ so $\widehat{\mu}_{n}^{\text{TADA}}$ minimizes $\ell_n^{\text{TADA}}$.
To proceed, we assume that $0<\tau_1\le\cdots\le\tau_m\le T< \tau_{m+1}\le\cdots \le\tau_n$, then Eq.\eqref{eqn: 7} becomes
\begin{equation}
\begin{aligned}
0&=\sum_{i=1}^m (\widehat{\mu}_{n}^{\text{TADA}}\tau_i-b)+\sum_{i=m+1}^n \frac{T}{\tau_i}(T\widehat{\mu}_{n}^{\text{TADA}}-b)\\
&=\Big(\sum_{i=1}^m\tau_i\Big)\widehat{\mu}_{n}^{\text{TADA}} -mb+\Big(\sum_{i=m+1}^n \frac{1}{\tau_i}\Big) T^2\widehat{\mu}_{n}^{\text{TADA}}-\Big(\sum_{i=m+1}^n \frac{1}{\tau_i}\Big)Tb\\
\end{aligned}
\end{equation}
which further simplifies to
\begin{equation}\label{eqn: 9}
\widehat{\mu}_{n}^{\text{TADA}}=\frac{mb+\left(\sum_{i=m+1}^n \frac{1}{\tau_i}\right)Tb}{\sum_{i=1}^m\tau_i+\left(\sum_{i=m+1}^n \frac{1}{\tau_i}\right)T^2}    
\end{equation}
Dividing both the numerator and denominator of Eq.\eqref{eqn: 9} by $n$, we can invoke the strong law of large numbers to get
\begin{equation}\label{eqn: 17}
\widehat{\mu}_{n}^{\text{TADA}}\xrightarrow{\text{a.s.}}\widetilde{\mu},\text{ where }\widetilde{\mu}\triangleq \frac{\mathbb{P}(\tau\le T)b+\mathbb{E}[\frac{T}{\tau}\mathbbm{1}_{\tau> T}]b}{\mathbb{E}[\tau\mathbbm{1}_{\tau\le T}]+\mathbb{E}[\frac{T}{\tau}\mathbbm{1}_{\tau> T}]T}
\end{equation}
Here $\xrightarrow{\text{a.s.}}$ denotes almost sure convergence, which is defined as $\mathbb{P}(\widehat{\mu}_{n}^{\text{TADA}}\rightarrow\widetilde{\mu})=1$.
Since $\mathbb{E}[\tau\mathbbm{1}_{\tau\le T}]<\mathbb{E}[T\mathbbm{1}_{\tau\le T}]=T\mathbb{P}(\tau\le T)$ and noting that the strict inequality follows from the fact that $\mathbb{P}(\tau<T)>0$, equation \eqref{eqn: 17} immediately gives $\widetilde{\mu}>\frac{b}{T}$. This shows that $\widehat{\mu}_{n}^{\text{TADA}}$ cannot be consistent as the true $\mu$ is allowed to be $\le \frac{b}{T}$. However, we can also prove the following stronger claim:
\begin{theorem}\label{thm: 2.2}
For any $\mu\in \mathbb{R}$, we always have $\widetilde{\mu}>\max\{\mu, \frac{b}{T}\}$.
\end{theorem}
To prove theorem \ref{thm: 2.2}, we need the following lemmas:
\begin{lemma}\label{lemma: 2.3}
If $0<a<A$ and $0<d<D$, then $\frac{A+d}{a+d}>\frac{A+D}{a+D}$.
\end{lemma}
The proof of lemma \ref{lemma: 2.3} is straightforward.
\begin{lemma}\label{lemma: 2.4}
$\varphi(x)\triangleq xe^{x^2}\operatorname{Erfc}(x)$ is strictly increasing on $\mathbb{R}$.
\end{lemma}
\begin{proof}[\proofname \text{ of Lemma \ref{lemma: 2.4}}]
$\varphi'(x)=e^{x^2}(2x^2+1)\operatorname{Erfc}(x)-\frac{2}{\sqrt{\pi}}x$. By the bounds on Mill's ratio $e^{x^2}\operatorname{Erfc}(x)>\frac{2}{\sqrt{\pi}}\frac{1}{x+\sqrt{x^2+2}}$ (see \cite{olver2010nist} 7.8.2), we have
\begin{equation*}
\begin{aligned}
\varphi'(x)&>\frac{2}{\sqrt{\pi}}\frac{2x^2+1}{x+\sqrt{x^2+2}}-\frac{2}{\sqrt{\pi}}x=\frac{2}{\sqrt{\pi}}\frac{x^2+1-x\sqrt{x^2+2}}{x+\sqrt{x^2+2}}\\
\end{aligned}
\end{equation*}
This is obviously $>0$ when $x\le0$. When $x>0$ we can continue to write
\begin{equation*}
\begin{aligned}
\varphi'(x)&>\frac{2}{\sqrt{\pi}}\frac{(x^2+1)^2-(x\sqrt{x^2+2})^2}{(x+\sqrt{x^2+2})(x^2+1+x\sqrt{x^2+2})}\\
&=\frac{2}{\sqrt{\pi}}\frac{1}{(x+\sqrt{x^2+2})(x^2+1+x\sqrt{x^2+2})}>0
\end{aligned}
\end{equation*}
So $\varphi(x)$ is strictly increasing on $\mathbb{R}$.
\end{proof}
\begin{proof}[\proofname \text{ of Theorem \ref{thm: 2.2}}]
Notice that $\mathbb{E}[\frac{T}{\tau}\mathbbm{1}_{\tau> T}]<\mathbb{E}[\mathbbm{1}_{\tau> T}]=\mathbb{P}(\tau>T)$
so by lemma \ref{lemma: 2.3} we have
\begin{equation*}
\begin{aligned}
 \widetilde{\mu}
=&\frac{b}{T}\frac{\mathbb{P}(\tau\le T)T+\mathbb{E}[\frac{T}{\tau}\mathbbm{1}_{\tau> T}]T}{\mathbb{E}[\tau\mathbbm{1}_{\tau\le T}]+\mathbb{E}[\frac{T}{\tau}\mathbbm{1}_{\tau> T}]T}\\
>&\frac{b}{T}\frac{\mathbb{P}(\tau\le T)T+\mathbb{P}(\tau>T)T}{\mathbb{E}[\tau\mathbbm{1}_{\tau\le T}]+\mathbb{P}(\tau>T)T}\\
=&\frac{b}{\mathbb{E}[\tau\mathbbm{1}_{\tau\le T}]+\mathbb{P}(\tau>T)T}\\ 
\end{aligned}
\end{equation*}
$\mathbb{P}(\tau> T)$ is given by Eq.\eqref{eqn: survival_prob1}. We can also compute that
\begin{equation}\label{eqn: 18}
\begin{aligned}
\mathbb{E}[\tau\mathbbm{1}_{\tau\le T}]=&\int_0^T \frac{bt}{\sqrt{2\pi t^3}}e^{-\frac{(b-\mu t)^2}{2t}} \diff t\\
=&\frac{b}{\mu}\left(1-\frac{1}{2}\operatorname{Erfc}\Big(\frac{T \mu-b}{\sqrt{2T}}\Big)-\frac{1}{2}e^{2 b \mu} \operatorname{Erfc}\Big(\frac{T \mu+b}{\sqrt{2T}}\Big)\right)\\
\end{aligned}
\end{equation}
(see details of this derivation in \ref{sec: deriv2}). So to prove that $\widetilde{\mu}>\mu$, it suffices to prove
\begin{equation*}
\tfrac{b}{\frac{b}{\mu}\left(1-\frac{1}{2}\operatorname{Erfc}\left(\frac{T \mu-b}{\sqrt{2T}}\right)-\frac{1}{2}e^{2 b \mu} \operatorname{Erfc}\left(\frac{T \mu+b}{\sqrt{2T}}\right)\right)+T\left(\frac{1}{2}\operatorname{Erfc}\left(\frac{T \mu-b}{\sqrt{2T}}\right)-\frac{1}{2}e^{2 b \mu} \operatorname{Erfc}\left(\frac{T \mu+b}{\sqrt{2T}}\right) \right)}>\mu\\ 
\end{equation*}
which further simplifies to
\begin{equation}\label{eqn: ineq}
(T\mu+b)e^{2 b \mu} \operatorname{Erfc}\Big(\frac{T \mu+b}{\sqrt{2T}}\Big)>(T\mu-b)\operatorname{Erfc}\Big(\frac{T \mu-b}{\sqrt{2T}}\Big).
\end{equation}
Note that this can be rewritten as $\varphi\Big(\frac{T\mu+b}{\sqrt{2T}}\Big)>\varphi\Big(\frac{T\mu-b}{\sqrt{2T}}\Big)$ where $\varphi(x)=xe^{x^2}\operatorname{Erfc}(x)$, so by Lemma \ref{lemma: 2.4}, the inequality \eqref{eqn: ineq} holds.
\end{proof}
% It suffices to prove \eqref{eqn: ineq} for $T\mu-b>0$. Let $\varphi(\mu,b, T)=(T\mu+b)e^{\mu b} \operatorname{Erfc}\Big(\frac{T \mu+b}{\sqrt{2T}}\Big)-(T\mu-b)e^{-\mu b} \operatorname{Erfc}\Big(\frac{T \mu-b}{\sqrt{2T}}\Big)$, we can compute
% \begin{equation*}
% \frac{\partial}{\partial T}\varphi(\mu,b, T)=e^{b \mu} \mu \operatorname{Erfc}\Big(\frac{T \mu+b}{\sqrt{2T}}\Big)-e^{-b \mu} \mu \operatorname{Erfc}\Big(\frac{T \mu-b}{\sqrt{2T}}\Big)
% \end{equation*}
% Let $g(\mu, b, T)=e^{-b \mu}  \operatorname{Erfc}\Big(\frac{T \mu-b}{\sqrt{2T}}\Big)-e^{b \mu}  \operatorname{Erfc}\Big(\frac{T \mu+b}{\sqrt{2T}}\Big)$ so that $\frac{\partial}{\partial T}\varphi(\mu,b, T)=-\mu g(\mu, b, T)$, since
% \begin{equation*}
% \frac{\partial}{\partial T}g(\mu,b, T)=-\sqrt{\frac{2}{\pi}}\frac{b e^{-\frac{b^2+T^2 \mu^2}{2 T}}}{T^{3 / 2}}<0
% \end{equation*}
% holds for all $\mu,b, T>0$, we know that $g$ is strictly decreasing w.r.t. $T$ and $g(\mu,b, T)>g(\mu, b, \infty)=0$, hence $\varphi$ is strictly decreasing w.r.t. $T$ for all $\mu,b>0$, so $\varphi(\mu, b, T)>\varphi(\mu, b, \infty)=0$.\qed

\section{Numerical Examples}
In this section we provide some numerical examples to further demonstrate the inconsistency of $\widehat{\mu}^{\text{TADA}}_n$.

\subsection{DDM with One-sided Boundary}\label{sec: exp3.1}
We first conduct a numerical illustration  of our proposed counterexample \eqref{eqn: model1} in Section \ref{sec: example}. We simulate an i.i.d. dataset of first passage times $(\tau_i)_{i=1}^n$ where $n=10, 000$, using Euler-Maruyama discretization of \eqref{eqn: model1} with a sufficiently small time step. With $b=1$ and $T=0.5$, we aim to infer $\mu$ from the dataset $(\tau_i)_{i=1}^n$ .

We compute both the TADA ``maximum likelihood'' estimator $\widehat{\mu}_n^\text{TADA}$ \eqref{eqn: glam_mle} and the generic maximum likelihood estimator $\widehat{\mu}_n^\text{ML}$ \eqref{eqn: addm_mle}. We then compare them with the true $\mu$ and $\widetilde{\mu}$ \eqref{eqn: 17}, the derived theoretical large-sample limit of $\widehat{\mu}_n^\text{TADA}$\footnote{$\mathbb{P}(\tau\le T)$ and $\mathbb{E}[\tau \mathbbm{1}_{\tau\le T}]$ are given by Eq.\eqref{eqn: survival_prob1} and Eq.\eqref{eqn: 18} respectively. $\mathbb{E}[\frac{T}{\tau}\mathbbm{1}_{\tau > T}]$ can be calculated with numerical integration accurately.}. The results are summarized in Table \ref{tab: example1}.
\begin{table}[htb!]
    \centering
    \begin{tabular}{cccc}
    \toprule
         $\mu$ & $\widehat{\mu}^\text{ML}_n$ & $\widetilde{\mu}$ &$\widehat{\mu}^{\text{TADA}}_n$ \\\midrule
         1& 0.996& 2.613&2.606\\
\bottomrule
    \end{tabular}
    \caption{Values of the true $\mu$, the maximum likelihood estimator $\widehat{\mu}_n^\text{ML}$ \eqref{eqn: addm_mle}, the TADA ``maximum likelihood'' estimator $\widehat{\mu}^\text{TADA}_n$ \eqref{eqn: glam_mle} and its limit $\widetilde{\mu}$ \eqref{eqn: 17}.}
    %$b$ and $T$ in \eqref{eqn: model1} are both set to be $1$.}
    \label{tab: example1}
\end{table}

The results demonstrate that $\widehat{\mu}_n^{\text{TADA}}$ converges to $\widetilde{\mu}$ rather than the true $\mu$ as $n \rightarrow \infty$, confirming its inconsistency. In contrast, the maximum likelihood estimator $\widehat{\mu}_n^{\mathrm{ML}}$ correctly converges to $\mu$.

\subsection{DDM with Double-sided Boundaries}\label{sec: exp3.2}

In this section we conduct a parameter recovery study for a DDM with one upper boundary and one lower boundary. The model is defined as
\begin{equation}\label{eqn: ddm_exp2}
\begin{aligned}
    &\diff X(t)=A(V(t)) \diff t+\sigma \diff W(t) ,\\
    \text{where }&A(v)\triangleq\begin{cases}
\mu_1& \text{if }v=\mathrm{A} ,\\
\mu_2& \text{if }v=\mathrm{B} .\\
\end{cases} 
\end{aligned}
\end{equation}
The drift function $\mu(t)=A(V(t))$ is a piecewise constant function alternating between $\mu_1$ and $\mu_2$ depending on the fixation trajectory $V(t)$. For each trial $i$, the fixation trajectory $V_i(t)$ was sampled from a gamma process to mimic the switching of attention in psychological experiments. Compared to the previous example in Section \ref{sec: exp3.1}, this model includes two absorbing boundaries and incorporates an attentional drift mechanism more closely aligned with the aDDM \eqref{eqn: addm}, while remaining a simplified variant that excludes option-value covariates.

We simulate a collection of $n=10,000$ visual fixation trajectories $(V_i)_{i=1}^n$ and the corresponding dataset $((\tau_i, C_i))_{i=1}^n$ from a prescribed DDM model \eqref{eqn: ddm_exp2} with alternating drifts $\mu_1, \mu_2$, constant diffusion $\sigma=1$ and symmetric constant boundaries $u=-\ell=1.5$ starting at $x_0=-0.2$. 
The average number of fixations per trial in our simulation was 3.806.

 TADA follows the same approach as in the one-boundary case: for each trial it replaces the time-varying drift rate $\mu(t)$ of the DDM with the constant drift rate $\tau^{-1}\int_0^\tau \mu(t)\diff t$. This constant drift is further used to evaluate the formulas for first passage time densities on the upper and lower boundaries of a standard DDM \eqref{eqn: ddm} with two boundaries, which are available as infinite series \citep{navarro2009fast, gondan2014even, cox2017theory}. This allows us to compute the estimator $\widehat{\mu}_{1, n}^\text{TADA}$ and $\widehat{\mu}_{2, n}^\text{TADA}$ from the simulated data. The computations of the generic MLEs $\widehat{\mu}_{1,n}^{\text{ML}}$ and $\widehat{\mu}_{2,n}^{\text{ML}}$, are however more involved, as the first passage time density of the DDM \eqref{eqn: ddm_exp2} is analytically intractable. In our parallel work \citep{liu2025efficient}, we develop an algorithm \textsc{efpt} for the fast and accurate computation of first passage time densities in DDMs with ``multi-stage'' structures, such as those with piecewise constant drift functions like \eqref{eqn: ddm_exp2} and \eqref{eqn: addm}. We use this algorithm to evaluate the likelihood and obtain the MLEs. The estimation results are summarized in Table \ref{tab: example2}.
\begin{table}[htb!]
    \centering
    \begin{tabular}{cccccc}
    \toprule
         $\mu_1$ & $\widehat{\mu}_{1, n}^\text{ML}$ &$\widehat{\mu}_{1,n}^{\text{TADA}}$ 
         &$\mu_2$ & $\widehat{\mu}_{2, n}^\text{ML}$ &$\widehat{\mu}_{2,n}^{\text{TADA}}$ \\\midrule
         $1$& $1.003$ & $1.684$ & $-0.8$& $-0.796$& $-1.478$\\
\bottomrule
    \end{tabular}
    \caption{Values of the true $\mu_1, \mu_2$, the maximum likelihood estimators $\widehat{\mu}_{1, n}^\text{ML}, \widehat{\mu}_{2, n}^\text{ML}$, the TADA ``maximum likelihood'' estimators $\widehat{\mu}^\text{TADA}_{1, n}, \widehat{\mu}^\text{TADA}_{2,n }$.} 
    %The boundaries are set to $u=-\ell=1.5$, and $x_0$ is set to $-0.2.$}
    \label{tab: example2}
\end{table}

From the table, we conclude again that $\widehat{\mu}_{1, n}^\text{ML}$ and $\widehat{\mu}_{2, n}^\text{ML}$ accurately estimate the true values of $\mu_1$ and $\mu_2$ respectively, while $\widehat{\mu}^\text{TADA}_{1, n}$ and $\widehat{\mu}^\text{TADA}_{2,n }$ do not. Similarly to the results in Example \ref{sec: exp3.1}, $\widehat{\mu}^\text{TADA}_{1, n}$ and $\widehat{\mu}^\text{TADA}_{2,n }$ overestimates the amplitude of the drifts. However, the theoretical large sample limits of $\widehat{\mu}^\text{TADA}_{1, n}$ and $\widehat{\mu}^\text{TADA}_{2,n }$ are unknown in this case.

\subsection{Attentional DDM}
We conclude this section by studying the consequences of using TADA to perform statistical inference of aDDM \eqref{eqn: addm} parameters. 

\subsubsection{Statistical Inconsistency in Parameter Recovery}

We first show the inconsistency in aDDM parameter estimation, as in the previous two subsections. Similar to the setup in Example~\ref{sec: exp3.2}, we simulate $n=10,000$ visual fixation trajectories $(V_i(t))_{i=1}^n$ from a gamma process. The corresponding option ratings $((r_{A, i}, r_{B, i}))_{i=1}^n$ are drawn from a discrete uniform distribution over $\{1,2,3,4,5\}$, and the associated decision data $((\tau_i, C_i))_{i=1}^n$ are generated accordingly. In this example, we set the boundaries to be symmetric ($u=-\ell$) and aim to estimate all four free parameters from data, namely $\eta, \kappa, u$ and $x_0$. The likelihood is then evaluated using both TADA (as described in Section~\ref{sec: glam}) and \textsc{efpt} \citep{liu2025efficient}, from which the corresponding maximum likelihood estimators are obtained.

The true parameters are set to $\kappa=0.5, u=2$, and $x_0=0.5$. Since $\eta$ is the primary parameter of interest for testing the aDDM, we vary its true value from 0 to 1 in increments of 0.05 to examine the estimation error across different levels of $\eta$. The relations of the true $\eta$, the TADA ``maximum likelihood'' estimator $\widehat{\eta}_n^{\text{TADA}}$, and the MLE $\widehat{\eta}_n^{\text{ML}}$ are shown in Figure \ref{fig: addm_param_recovery}.
\begin{figure}[htbp]
  \centering
\begin{minipage}{0.49\textwidth}
  \includegraphics[width=\linewidth]{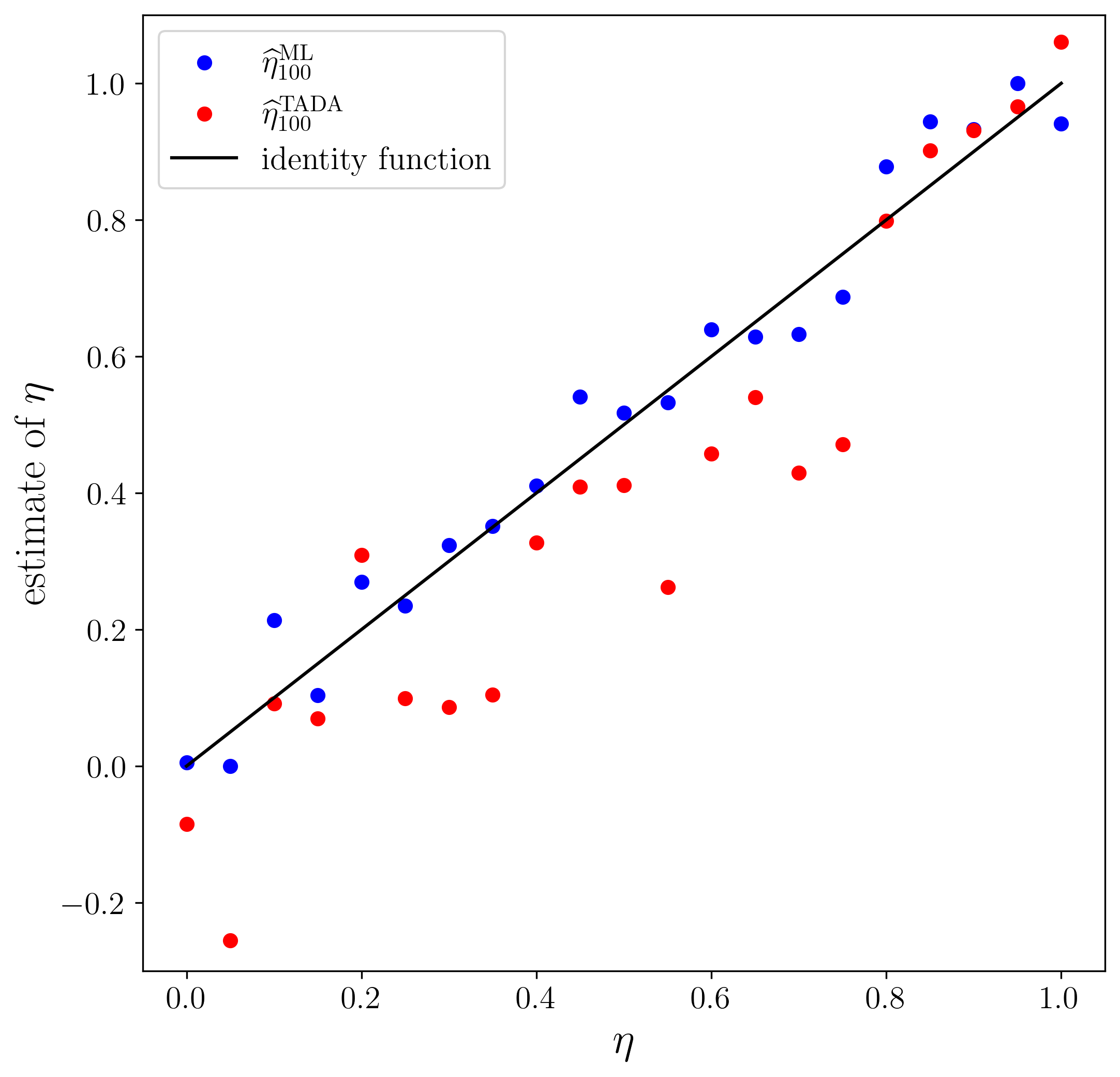}
	\end{minipage}
\begin{minipage}{0.49\textwidth}
  \includegraphics[width=\linewidth]{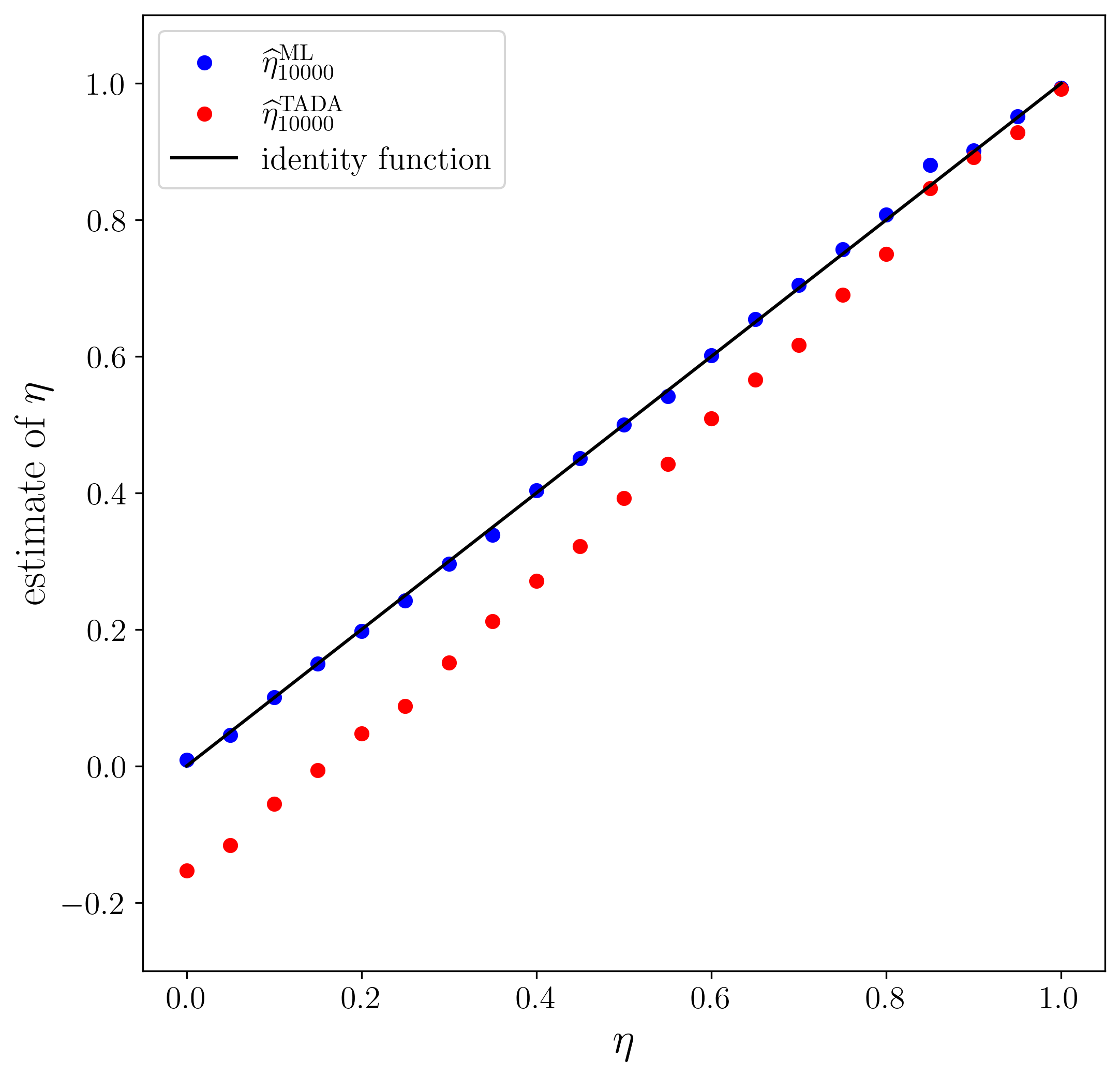}
	\end{minipage}
  \caption{Plots of the values of estimator $\widehat{\eta}_n^{\text{TADA}}$ (red points) and $\widehat{\eta}_n^{\text{ML}}$ (blue points) against the true value $\eta$ with $n=100$ (left subplot) and $n=10,000$ (right subplot). The black line is the plot of the identity function $y=x$ which serves as the reference for consistent estimation. $\eta=0$ means maximal attentional effect that the decision-maker completely ignores the non-fixated item in evidence accumulation, whereas $\eta=1$ corresponds to no attentional effect, in which case the model reduces to the standard DDM.}
  \label{fig: addm_param_recovery}
\end{figure}

From Figure \ref{fig: addm_param_recovery}, we can see that the $\widehat{\eta}_n^{\text{ML}}$ aligns closely with the identity line $y=x$, even for small sample size ($n=100$), and becomes nearly exact when $n=10,000$. This indicates the consistency of $\widehat{\eta}_n^{\text{ML}}$: as the sample size grows, the estimator converges to the true parameter. However, $\widehat{\eta}_n^{\text{TADA}}$ systematically underestimates $\eta$. Even when $n=10,000$, the red points remain well below the identity line, showing persistent asymptotic bias that does not vanish with larger $n$. This confirms that the TADA ``maximum likelihood'' estimator is inconsistent. It is worth noting that $\eta=1$ corresponds to the case with no attentional effect, where the model reduces to the standard DDM. In this case, the TADA ``likelihood'' coincides with the analytical likelihood solution of the standard DDM, and the bias therefore vanishes. As $\eta$ decreases toward 0 (indicating stronger attentional effects), the bias $\widehat{\eta}_n^{\text{TADA}}-\eta$ grows larger. Indeed, the original aDDM study \citep{krajbich2010visual}, which fitted the model via expensive forward simulation, reported a maximum-likelihood estimate of $0.3$, a value in the range where TADA exhibits substantial bias. When $\eta$ is close to 0 (approximately $< 0.2), \widehat{\eta}_n^{\text{TADA}}$ even takes negative values, which are not physically meaningful\footnote{The negative values of \(\widehat{\eta}_n^{\mathrm{TADA}}\) arise partly because we use unconstrained numerical optimization in this experiment. If one instead imposes the natural constraint \(\eta\in[0,1]\), then estimates that would otherwise be negative are truncated at the boundary \(\eta=0\). This keeps the estimates within the admissible parameter space, but it does not resolve the underlying inconsistency of TADA.}. Indeed, a recent study comparing TADA fits to multiple datasets across humans and monkeys often reported individuals with $\eta < 0$ \citep{lupkin2023monkeys}.

In this specific example, the TADA estimator is strongly correlated with the true parameter, perhaps suggesting that it can still be useful for qualitatively measuring attentional effects in an aDDM. Indeed, our results need not imply that all previous reported qualitative conclusions about attentional effects based on TADA are invalid. However, in the next section, we show that it is easy to construct realistic circumstances in which TADA can lead to erroneous qualitative scientific conclusions about which of two populations exhibits a stronger attentional effect, providing further caution against its use. The problem is that the magnitude of the bias exhibited in Figure \ref{fig: addm_param_recovery} is affected by the other parameters defining the aDDM. If these parameters differ across populations, then differences in TADA estimates of attention can be driven by changes in these other parameters, not by changes in attention.

% We note that, in this specific example, if researchers are primarily interested in qualitative conclusions of the aDDM rather than quantitative parameter estimation, the use of TADA may still yield meaningful insights. Scientifically, Figure \ref{fig: addm_param_recovery} shows that, at least in this specific example, while TADA underestimates the magnitude of the attentional parameter $\eta$, it nonetheless preserves the monotonic relationship between the estimated and true values. This means that TADA still captures the direction and relative strength of attentional effects: trials with stronger true attentional modulation (smaller $\eta$) still produce smaller estimated values, although this monotonic relationship may only become evident with a large dataset (e.g., $n=10,000$). However, we caution that this monotonicity relationship is specific to the present simulation and may attenuate or even fail under different task structures, model variants, or parameter regimes, so even the qualitative conclusions based on TADA may not hold in other settings. Indeed, in the next part, we show that hypothesis testing based on TADA can fail in a much more fundamental way: it could lead to incorrect scientific conclusions.

\subsubsection{Scientific Hypothesis Testing}

Let \RomanNumeralCaps{1} and \RomanNumeralCaps{2} denote two populations, and let \(\eta_\RomanNumeralCaps{1}\) and \(\eta_\RomanNumeralCaps{2}\) denote their respective attentional modulation parameters. A natural scientific question is whether the two populations differ in how attention shapes their decisions, or whether attentional modulation is stronger in one population than in the other. This question can be formalized by comparing the attentional modulation parameters \(\eta_\RomanNumeralCaps{1}\) and \(\eta_\RomanNumeralCaps{2}\) using decision data collected from the two populations. The inferential target is the contrast
\begin{equation*}
\Delta \eta \triangleq \eta_\RomanNumeralCaps{1}-\eta_\RomanNumeralCaps{2} .
\end{equation*}

We consider two plug-in estimators of this contrast:
\begin{equation*}
\begin{aligned}
\widehat{\Delta\eta}^{\mathrm{ML}}
&\triangleq
\widehat{\eta}_\RomanNumeralCaps{1}^{\mathrm{ML}}
-
\widehat{\eta}_\RomanNumeralCaps{2}^{\mathrm{ML}},\\
\widehat{\Delta\eta}^{\mathrm{TADA}}
&\triangleq
\widehat{\eta}_\RomanNumeralCaps{1}^{\mathrm{TADA}}
-
\widehat{\eta}_\RomanNumeralCaps{2}^{\mathrm{TADA}},
\end{aligned}
\end{equation*}
where \(\widehat{\eta}_\RomanNumeralCaps{1}^{\mathrm{ML}}\) and \(\widehat{\eta}_\RomanNumeralCaps{2}^{\mathrm{ML}}\) are obtained by maximizing the full aDDM likelihood (computed by \textsc{efpt} \cite{liu2025efficient}) separately for the two populations, while \(\widehat{\eta}_\RomanNumeralCaps{1}^{\mathrm{TADA}}\) and \(\widehat{\eta}_\RomanNumeralCaps{2}^{\mathrm{TADA}}\) are the corresponding estimates obtained under the TADA approximation.

Because the two samples come from different populations, it is important not to assume that the remaining model components are identical. The nuisance parameters, such as drift scale $\kappa$, boundary separation $a$, and starting point $x_0$, may differ across populations. Moreover, the covariates themselves, such as item ratings and fixation patterns, may follow different distributions in the two populations. 

We consider two simulation settings.
\begin{enumerate}
    \item Populations \RomanNumeralCaps{1} and \RomanNumeralCaps{2} share the same nuisance parameters \(\kappa=0.5\), \(a=2\), and \(x_0=0\), and their fixation times are sampled from the same gamma distribution with shape 4 and rate 10. Their item ratings differ: population \RomanNumeralCaps{1} uses \(\mathrm{Uniform}(\{1,\ldots,9\})\), while population \RomanNumeralCaps{2} uses \(\mathrm{Uniform}(\{4,5,6\})\). These distributions have the same mean 5 but different variability. This setting captures a population difference in the dispersion of subjective item values.
    
    \item Populations \RomanNumeralCaps{1} and \RomanNumeralCaps{2} share the same parameters \(\kappa=0.5\) and \(x_0=0\). Their fixation times are sampled from the same gamma distribution with shape 4 and rate 10, and their item ratings are both sampled from \(\mathrm{Uniform}(\{1,2,3\})\). However, the boundary parameters differ: \(a_{\RomanNumeralCaps{1}}=1.2\) and \(a_{\RomanNumeralCaps{2}}=3.5\). The second setting captures a population difference in response caution.
\end{enumerate}
For both settings, we set \(\eta_{\RomanNumeralCaps{1}}=0.3\) and estimate \(\Delta\eta\) as the true \(\Delta\eta\) varies from $-0.1$ to $0.1$. We plot the estimators $\widehat{\Delta\eta}$ from $n=10,000$ number of data versus the true $\Delta\eta$ in Figure \ref{fig: testing}.

\begin{figure}[htbp]
  \centering
\begin{minipage}{0.49\textwidth}
  \includegraphics[width=\linewidth]{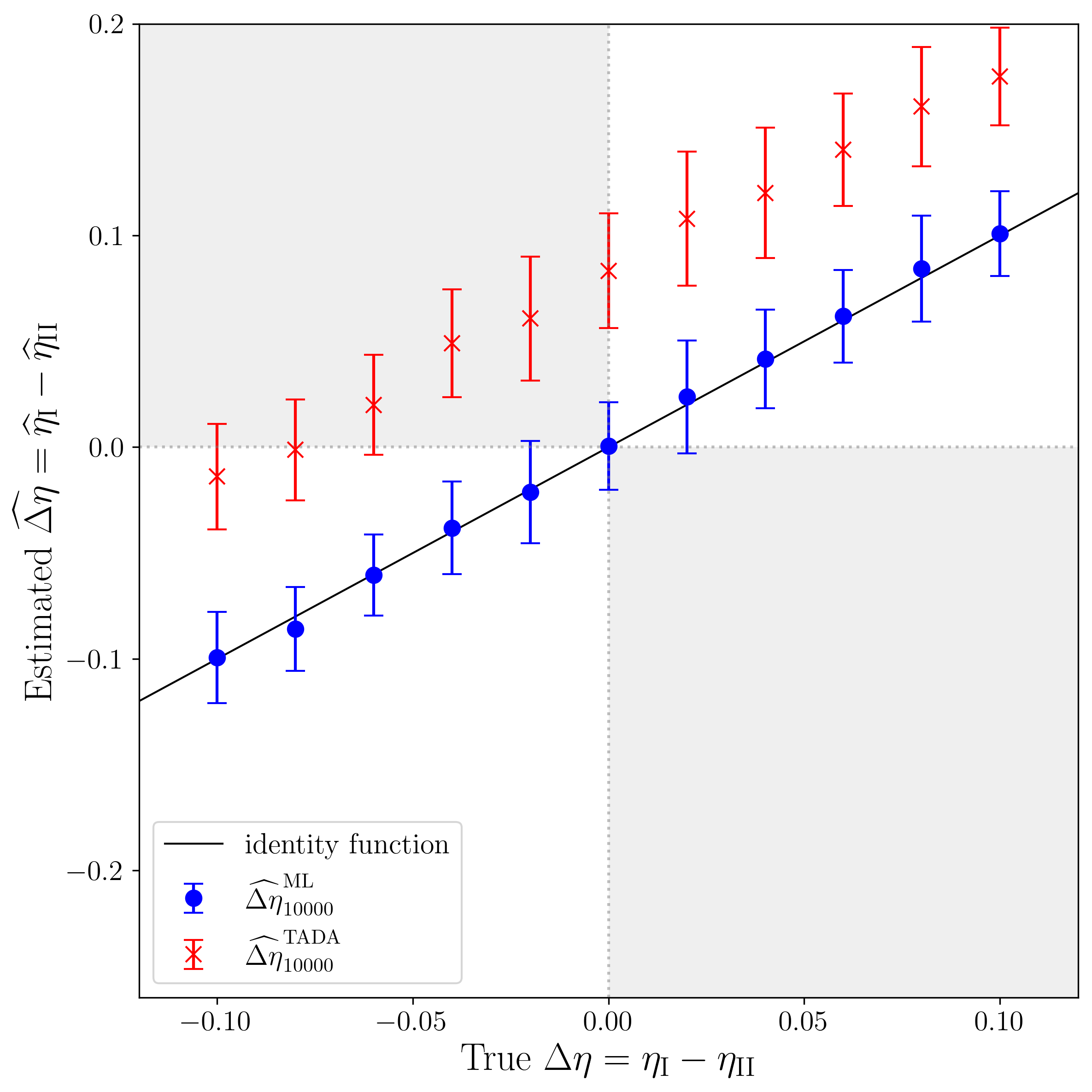}
	\end{minipage}
\begin{minipage}{0.49\textwidth}
  \includegraphics[width=\linewidth]{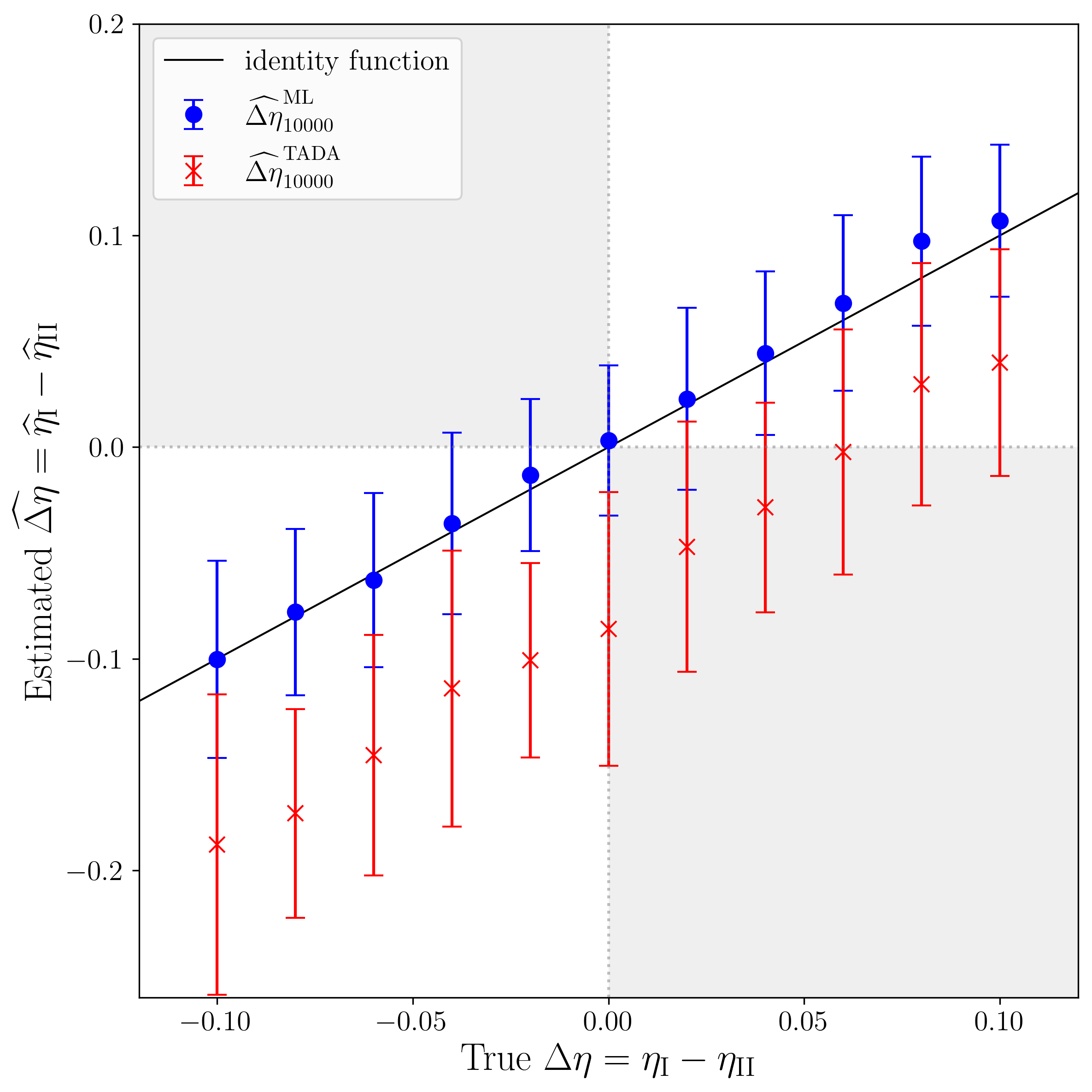}
	\end{minipage}
  \caption{Plots of the values of estimator $\widehat{\Delta\eta}_n^{\text{TADA}}$ (red points) and $\widehat{\Delta\eta}_n^{\text{ML}}$ (blue points) against the true value $\eta$ with $n=10000$ under two different simulation settings. In the left panel, populations \RomanNumeralCaps{1} and \RomanNumeralCaps{2} differ in their item-rating distributions. In the right panel, they differ in their boundary parameters. The black line is the plot of the identity function $y=x$ which serves as the reference for consistent estimation. Error bars represent the standard deviations across 50 replicate random simulations. The second and fourth quadrants are shaded, indicating regions where the estimated $\Delta \eta$ has the opposite sign from the true $\Delta \eta$. Thus, any estimate falling in these shaded regions corresponds to an incorrect qualitative conclusion about which group has the larger attentional effect.}
  \label{fig: testing}
\end{figure}

Figure \ref{fig: testing} shows that, unsurprisingly, $\widehat{\Delta \eta}^{\mathrm{TADA}}$ is inconsistent in both settings. More concerningly, the bias can be large enough to reverse the sign of the estimated contrast relative to the true $\Delta \eta$. For example, in the left panel, when $\Delta \eta$ is approximately in $(-0.08,0), \widehat{\Delta \eta}_{10000}^{\text{TADA}}$ is positive. In the right panel, when $\Delta \eta$ is approximately in $(0,0.06), \widehat{\Delta \eta}_{10000}^{\text{TADA}}$ mostly takes negative values. Thus, TADA may not merely distort the magnitude of the estimated attentional effect; it can also lead to qualitatively incorrect conclusions about which population exhibits stronger attentional modulation.

From the results in Figure \ref{fig: testing}, we can formally construct hypothesis testing examples to showcase how TADA inconsistency can mislead scientific conclusions. Specifically,
\begin{enumerate}
    \item Under simulation setting 1, we take \(\Delta\eta=-0.05\), hence \(\eta_{\RomanNumeralCaps{1}}=0.30\) and \(\eta_{\RomanNumeralCaps{2}}=0.35\). We conduct the following one-sided likelihood ratio test (LRT)\footnote{As discussed earlier in Section \ref{sec: non-equivalence}, the TADA FPTD does not yield a valid likelihood function, and the likelihood ratio test based on the TADA likelihood should therefore only be understood in a formal sense.}:
    \begin{equation*}
        H_0: \eta_{\RomanNumeralCaps{1}} \le \eta_{\RomanNumeralCaps{2}}
        \quad\text{vs.}\quad
        H_1: \eta_{\RomanNumeralCaps{1}} > \eta_{\RomanNumeralCaps{2}},
    \end{equation*}
    at significance level \(\alpha=0.05\). In this setting, \(H_0\) is true.

    We repeat the simulation 200 times. The LRT based on the aDDM likelihood correctly does not reject \(H_0\) in any replication, resulting in no Type \RomanNumeralCaps{1} errors. In contrast, the LRT based on the TADA likelihood incorrectly rejects \(H_0\) 69 times, corresponding to an empirical Type \RomanNumeralCaps{1} error rate of \(69/200=0.345\).

    \item Under simulation setting 2, we take \(\Delta\eta=0.05\), hence \(\eta_{\RomanNumeralCaps{1}}=0.30\) and \(\eta_{\RomanNumeralCaps{2}}=0.25\). We conduct the same one-sided LRT at significance level \(\alpha=0.05\). In this setting, \(H_0\) is false.

    Again, we repeat the simulation 200 times. The LRT based on the aDDM likelihood rejects \(H_0\) 72 times, corresponding to an empirical power of \(72/200=0.36\). In contrast, the LRT based on the TADA likelihood rejects \(H_0\) only 2 times, corresponding to an empirical power of \(2/200=0.01\).
\end{enumerate}

The reason behind these results is simple: the non-zero asymptotic bias for TADA ``maximum likelihood'' estimator varies with nuisance parameter and covariate values. Therefore, when comparing two populations with different nuisance parameters or different covariate distributions, the TADA-induced biases in the two populations need not cancel, and can distort the relative size. Thus, TADA can fail not only quantitatively, by producing biased parameter estimates, but also qualitatively, by leading to incorrect scientific conclusions about the direction of population differences, especially when these differences are subtle, as may often be the case.

These are only two representative examples of the consequences of TADA inconsistency; more failures should be expected in other settings. In contrast, the generic maximum likelihood estimator is consistent, so its asymptotic bias is always zero, and the estimated values closely track the identity line across the range of $\Delta\eta$. It fits naturally within the likelihood-ratio testing framework and produces statistically interpretable results.

% In practice, this implies that GLAM can qualitatively distinguish between experimental conditions that differ in the extent of gaze bias or attentional influence, even though the estimated numerical values of $\eta$ are not accurate in an absolute sense. Thus, while GLAM should not be relied upon for precise parameter inference or hypothesis testing, it can still provide an approximate descriptive summary of how gaze modulates evidence accumulation in decision making.
% Multiplicative:
% \begin{equation*}
% M(v, r_A, r_B)=\begin{cases}
% 			\kappa(r_A-\eta r_B) & \text{if }v=\mathrm{A} ,\\
% 			\kappa(r_B-\eta r_A)\textcolor{blue}{\text{or  } \kappa(\eta r_A-r_B)???} & \text{if }v=\mathrm{B} ,\\
% 		\end{cases}
% \end{equation*}
% Additive:
% \begin{equation*}
% A(v, r_A, r_B)=\begin{cases}
% 			\kappa(r_A-r_B+\beta) & \text{if }v=\mathrm{A} ,\\
% 			\kappa(r_A-r_B-\beta) & \text{if }v=\mathrm{B} ,\\
% 		\end{cases}
% \end{equation*}

\section{Conclusion}
In this paper, we present several simple yet compelling examples demonstrating that ``maximum likelihood'' estimators based on TADA yield inconsistent estimates of the aDDM parameters. This inconsistency is a natural consequence of the fact that TADA does not define a well-specified forward generative model (because the TADA drift at a given time is only determined by the {\em future} decision time).   
%(using relative fixation duration as a covariate is post-hoc, one needs to observe the corresponding reaction time first to extract this summary). 
We establish this inconsistency through theoretical derivations for the DDM with a single absorbing boundary, support it with numerical simulations, and further demonstrate it numerically for the DDM with two absorbing boundaries, as well as for the full aDDM. For the aDDM we also provide an example illustrating the deleterious effects of inconsistency on downstream scientific conclusions. 

Numerous recent studies on the aDDM have relied on TADA or similar approximations that ignore the fixation path \citep{lupkin2023monkeys, sepulveda2020visual, ramirez2022optimal, yang2023dynamic, yang2024attention, ting2024unraveling, zilker2024attentional, thomas2019gaze, thomas2021uncovering, cavanagh2014eye, smith2019estimating}. %Since TADA is based on an intuitive approximation rather than a rigorous mathematical foundation, there is no general theoretical guarantee that this ``approximation'' is reliable across settings. 
This paper suggests that such an approach may not be as reliable as previously assumed, owing to the demonstrated inconsistency of the resulting estimators. We consider our findings as a note of caution in that regard. 

% More importantly, when TADA is used for downstream scientific analyses, this inconsistency can lead to incorrect substantive conclusions. 

We note that although this work focuses on the canonical multiplicative form of the aDDM, our main conclusion regarding estimator inconsistency extends to other model variants as well, as long as they also admit a similar TADA type estimator described in Section \ref{sec: glam}. More broadly, attention-driven evidence accumulation for preferential choice has a long history in the sequential-sampling tradition, from decision field theory and its multi-alternative extension \citep{busemeyer1993decision, roe2001multialternative} to models in which the drift is explicitly piecewise and switches with attention, such as multistage diffusion models \citep{diederich2016multistage} and multi-attribute accumulator models \citep{holmes2016new, trueblood2014multiattribute}.
While these models vary in details, these capture the notion that decision makers can selectively attend to different attributes of each choice option, constructing their preferences for one option over the other via sequential sampling of momentary preferences. Typically, such models have either not been fit quantitatively to experimental data \citep{roe2001multialternative}, relied on TADA approaches \citep{yang2023dynamic} or expensive simulations \citep{trueblood2026}. Furthermore, beyond the attention effect, constant-parameter approximations to DDMs with time-dependent drift terms, where the ``effective'' constant parameter depends on the first passage time itself, should generally not be expected to yield reliable statistical conclusions.

Our findings underscore the need for a fast and statistically sound algorithm for computing the likelihood function of the aDDM, either exactly or through mathematically justified approximations. Our parallel work \citep{liu2025efficient} introduces such a fast and accurate algorithm \textsc{efpt} for evaluating first passage time densities for DDMs, achieving state-of-the-art inference speed and outperforming previous methods \citep{shinn2020flexible} by orders of magnitude for aDDMs. In fact, this work enabled the efficient execution of the numerical experiments presented in this paper. We therefore recommend that future aDDM analyses use likelihood-based methods such as those developed in \citet{liu2025efficient} and the accompanying software in \href{https://github.com/RiverFlowsInYou98/efficient-fpt}{https://\allowbreak github.com/\allowbreak RiverFlowsInYou98/\allowbreak efficient-fpt}, and that existing results obtained using TADA be reexamined when appropriate. The algorithms in \citet{liu2025efficient} are grounded in rigorous mathematical theory while simultaneously offering computational efficiency comparable to TADA. We expect that the core ideas underlying the algorithms of \citet{liu2025efficient} can be extended to the model classes reviewed in the preceding paragraph. Developments in this direction would provide statistically grounded alternatives to TADA-type or simulation-based fitting procedures much more broadly. More generally, we view the development of fast and reliable likelihood algorithms for increasingly general classes of DDMs as an important direction for future work.

% \subsubsection*{Acknowledgements}
% This work was supported by the National Institute of Mental Health (grants P50 MH119467-01 and P50 MH106435-06A1), the Office of Naval Research (MURI Award N00014-23-1-2792), and the Brainstorm Program at the Robert J. and Nancy D. Carney Institute for Brain Science. S.L. was additionally supported by the JBB Endowed Graduate Fellowship in Brain Science.

\section*{Declarations}
\subsection*{Acknowledgements}
This work was supported by the National Institute of Mental Health (grants P50 MH119467-01 and P50 MH106435-06A1), the Office of Naval Research (MURI Award N00014-23-1-2792), and the Brainstorm Program at the Robert J. and Nancy D. Carney Institute for Brain Science. S.L. was additionally supported by the JBB Endowed Graduate Fellowship in Brain Science.

\subsection*{Contributions}
M.J.F. and A.F. raised the research question and provided expertise on the computational cognitive science applications. M.T.H. and S.L. developed the methodological and theoretical framework. S.L. carried out the research, validated the results, and prepared the associated GitHub repository with support from all authors. S.L. wrote the original draft with feedback and revisions from all authors. M.J.F. secured funding for the project. M.T.H. supervised the project.

\subsection*{Data Availability}

The code supporting this study will be made publicly available at \href{https://github.com/RiverFlowsInYou98/TADA}{https://\allowbreak github.com/\allowbreak RiverFlowsInYou98/TADA} upon publication.

\subsection*{Competing interests}

The authors declare no competing interests.

%\newpage
\bibliographystyle{elsarticle-harv}
\bibliography{bibtex}

@article{ging-jehli_basal_ganglia_2025,
  title = {Basal ganglia components have distinct computational roles in decision-making dynamics under conflict and uncertainty},
  volume = {23},
  ISSN = {1545-7885},
  DOI = {10.1371/journal.pbio.3002978},
  number = {1},
  journal = {PLOS Biology},
  publisher = {Public Library of Science (PLoS)},
  author = {Ging-Jehli,  Nadja and Cavanagh,  James F. and Ahn,  Minkyu and Segar,  David J. and Asaad,  Wael F. and Frank,  Michael J.},
  editor = {Kaplan,  Raphael Samuel Matthew},
  year = {2025},
  pages = {e3002978}
}

@article{Drugowitsch2012,
  title = {The Cost of Accumulating Evidence in Perceptual Decision Making},
  volume = {32},
  ISSN = {1529-2401},
  url = {http://dx.doi.org/10.1523/JNEUROSCI.4010-11.2012},
  DOI = {10.1523/jneurosci.4010-11.2012},
  number = {11},
  journal = {The Journal of Neuroscience},
  publisher = {Society for Neuroscience},
  author = {Drugowitsch,  Jan and Moreno-Bote,  Rubén and Churchland,  Anne K. and Shadlen,  Michael N. and Pouget,  Alexandre},
  year = {2012},
  month = mar,
  pages = {3612–3628}
}

@article{shinn2020flexible,
  title={A flexible framework for simulating and fitting generalized drift-diffusion models},
  author={Shinn, Maxwell and Lam, Norman H and Murray, John D},
  journal={ELife},
  volume={9},
  pages={e56938},
  year={2020},
  publisher={eLife Sciences Publications, Ltd}
}

@article{krajbich2010visual,
  title={Visual fixations and the computation and comparison of value in simple choice},
  author={Krajbich, Ian and Armel, Carrie and Rangel, Antonio},
  journal={Nature neuroscience},
  volume={13},
  number={10},
  pages={1292--1298},
  year={2010},
  publisher={Nature Publishing Group US New York}
}

@article{ratcliff2008,
  title={The diffusion decision model: theory and data for two-choice decision tasks},
  author={Ratcliff, Roger and McKoon, Gail},
  journal={Neural computation},
  volume={20},
  number={4},
  pages={873--922},
  year={2008},
  publisher={MIT Press}
}

@article{ratcliff2016,
title = {Diffusion Decision Model: Current Issues and History},
journal = {Trends in Cognitive Sciences},
volume = {20},
number = {4},
pages = {260-281},
year = {2016},
issn = {1364-6613},
doi = {https://doi.org/10.1016/j.tics.2016.01.007},
url = {https://www.sciencedirect.com/science/article/pii/S1364661316000255},
author = {Roger Ratcliff and Philip L. Smith and Scott D. Brown and Gail McKoon}
}

@article{Peters2019,
title={The drift diffusion model as the choice rule in inter-temporal and risky choice: A case study in medial orbitofrontal cortex lesion patients and controls},
author={J. Peters and M. D’Esposito},
journal={PLoS Computational Biology},
year={2019},
volume={16},
doi={10.1101/642587}
}

@article{Clithero2016Response,
title={Response Times in Economics: Looking Through the Lens of Sequential Sampling Models},
author={J. Clithero},
journal={CSN: General Cognitive Neuroscience (Topic)},
year={2016},
doi={10.2139/ssrn.2795871}
}

@article{Myers2022A,
title={A practical introduction to using the drift diffusion model of decision-making in cognitive psychology, neuroscience, and health sciences},
author={C. Myers and A. Interian and A. Moustafa},
journal={Frontiers in Psychology},
year={2022},
volume={13},
doi={10.3389/fpsyg.2022.1039172}
}

@article{Pedersen2017The,
title={The drift diffusion model as the choice rule in reinforcement learning},
author={Mads L. Pedersen and M. Frank and G. Biele},
journal={Psychonomic Bulletin \& Review},
year={2017},
volume={24},
pages={1234-1251},
doi={10.3758/s13423-016-1199-y}
}

@article{milosavljevic2010drift,
  title={The drift diffusion model can account for value-based choice response times under high and low time pressureThe drift diffusion model can account for value-based choice response times under high and low time pressure},
  author={Milosavljevic, Milica and Malmaud, Jonathan and Huth, Alexander and Koch, Christof and Rangel, Antonio},
  journal={Judgment and Decision Making},
  volume={5},
  number={6},
  pages={437--449},
  year={2010},
  publisher={Society for Judgment and Decision Making}
}

@article{cisek2009decisions,
  title={Decisions in changing conditions: the urgency-gating model},
  author={Cisek, Paul and Puskas, Genevi{\`e}ve Aude and El-Murr, Stephany},
  journal={Journal of Neuroscience},
  volume={29},
  number={37},
  pages={11560--11571},
  year={2009},
  publisher={Soc Neuroscience}
}

@article{Turner2015Informing,
title={Informing cognitive abstractions through neuroimaging: the neural drift diffusion model.},
author={Brandon M. Turner and Leendert van Maanen and B. Forstmann},
journal={Psychological review},
year={2015},
volume={122 2},
pages={ 312-336 },
doi={10.1037/a0038894}
}

@article{fengler2022beyond,
  title={Beyond Drift Diffusion Models: Fitting a Broad Class of Decision and Reinforcement Learning Models with HDDM},
  author={Fengler, Alexander and Bera, Krishn and Pedersen, Mads L and Frank, Michael J},
  journal={Journal of Cognitive Neuroscience},
  volume={34},
  number={10},
  pages={1780--1805},
  year={2022},
  publisher={MIT Press One Broadway, 12th Floor, Cambridge, Massachusetts 02142, USA~…}
}

@article{smith2019estimating,
  title={Estimating the dynamic role of attention via random utility},
  author={Smith, Stephanie M and Krajbich, Ian and Webb, Ryan},
  journal={Journal of the Economic Science Association},
  volume={5},
  number={1},
  pages={97--111},
  year={2019},
  publisher={Cambridge University Press \& Assessment}
}

@article{busemeyer1993decision,
  title={Decision field theory: a dynamic-cognitive approach to decision making in an uncertain environment.},
  author={Busemeyer, Jerome R and Townsend, James T},
  journal={Psychological review},
  volume={100},
  number={3},
  pages={432},
  year={1993},
  publisher={American Psychological Association}
}

@article{roe2001multialternative,
  title={Multialternative decision field theory: A dynamic connectionst model of decision making.},
  author={Roe, Robert M and Busemeyer, Jermone R and Townsend, James T},
  journal={Psychological review},
  volume={108},
  number={2},
  pages={370},
  year={2001},
  publisher={American Psychological Association}
}

@article{diederich2016multistage,
  title={A multistage attention-switching model account for payoff effects on perceptual decision tasks with manipulated processing order.},
  author={Diederich, Adele},
  journal={Decision},
  volume={3},
  number={2},
  pages={81},
  year={2016},
  publisher={Educational Publishing Foundation}
}

@article{holmes2016new,
  title={A new framework for modeling decisions about changing information: The Piecewise Linear Ballistic Accumulator model},
  author={Holmes, William R and Trueblood, Jennifer S and Heathcote, Andrew},
  journal={Cognitive psychology},
  volume={85},
  pages={1--29},
  year={2016},
  publisher={Elsevier}
}

@article{trueblood2014multiattribute,
  title={The multiattribute linear ballistic accumulator model of context effects in multialternative choice.},
  author={Trueblood, Jennifer S and Brown, Scott D and Heathcote, Andrew},
  journal={Psychological review},
  volume={121},
  number={2},
  pages={179},
  year={2014},
  publisher={American Psychological Association}
}

@article{wiecki2013hddm,
  title={HDDM: Hierarchical Bayesian estimation of the drift-diffusion model in Python},
  author={Wiecki, Thomas V and Sofer, Imri and Frank, Michael J},
  journal={Frontiers in neuroinformatics},
  volume={7},
  pages={14},
  year={2013},
  publisher={Frontiers}
}

@article{malhotra2018,
  title={Time-varying decision boundaries: insights from optimality analysis},
  author={Malhotra, Gaurav and Leslie, David S and Ludwig, Casimir JH and Bogacz, Rafal},
  journal={Psychonomic bulletin \& review},
  volume={25},
  number={3},
  pages={971--996},
  year={2018},
  publisher={Springer}
}

@article{palestro2018,
  title={Some task demands induce collapsing bounds: Evidence from a behavioral analysis},
  author={Palestro, James J and Weichart, Emily and Sederberg, Per B and Turner, Brandon M},
  journal={Psychonomic bulletin \& review},
  volume={25},
  number={4},
  pages={1225--1248},
  year={2018},
  publisher={Springer}
}

@article{wieschen2020jumping,
  title={Jumping to conclusion? a l{\'e}vy flight model of decision making},
  author={Wieschen, Eva Marie and Voss, Andreas and Radev, Stefan},
  journal={TQMP},
  volume={16},
  number={2},
  pages={120--132},
  year={2020}
}

@article{rasanan2024there,
  title={Are there jumps in evidence accumulation, and what, if anything, do they reflect psychologically? An analysis of L{\'e}vy Flights models of decision-making},
  author={Rasanan, Amir Hosein Hadian and Rad, Jamal Amani and Sewell, David K},
  journal={Psychonomic Bulletin \& Review},
  volume={31},
  number={1},
  pages={32--48},
  year={2024},
  publisher={Springer}
}

@article{rasanan2023numerical,
  title={Numerical approximation of the first-passage time distribution of time-varying diffusion decision models: A mesh-free approach},
  author={Rasanan, Amir Hosein Hadian and Evans, Nathan J and Rieskamp, J{\"o}rg and Rad, Jamal Amani},
  journal={Engineering Analysis with Boundary Elements},
  volume={151},
  pages={227--243},
  year={2023},
  publisher={Elsevier}
}

@article{richter2023diffusion,
	title={Diffusion models with time-dependent parameters: An analysis of computational effort and accuracy of different numerical methods},
	author={Richter, Thomas and Ulrich, Rolf and Janczyk, Markus},
	journal={Journal of Mathematical Psychology},
	volume={114},
	pages={102756},
	year={2023},
	publisher={Elsevier}
}

@article{cavanagh2014eye,
  title={Eye tracking and pupillometry are indicators of dissociable latent decision processes.},
  author={Cavanagh, James F and Wiecki, Thomas V and Kochar, Angad and Frank, Michael J},
  journal={Journal of Experimental Psychology: General},
  volume={143},
  number={4},
  pages={1476},
  year={2014},
  publisher={American Psychological Association}
}

@article{lupkin2023monkeys,
  title={Monkeys exhibit human-like gaze biases in economic decisions},
  author={Lupkin, Shira M and McGinty, Vincent B},
  journal={Elife},
  volume={12},
  pages={e78205},
  year={2023},
  publisher={eLife Sciences Publications Limited}
}

@article{sepulveda2020visual,
  title={Visual attention modulates the integration of goal-relevant evidence and not value},
  author={Sepulveda, Pradyumna and Usher, Marius and Davies, Ned and Benson, Amy A and Ortoleva, Pietro and De Martino, Benedetto},
  journal={elife},
  volume={9},
  pages={e60705},
  year={2020},
  publisher={eLife Sciences Publications, Ltd}
}

@article{ramirez2022optimal,
  title={Optimal allocation of finite sampling capacity in accumulator models of multialternative decision making},
  author={Ram{\'\i}rez-Ruiz, Jorge and Moreno-Bote, Rub{\'e}n},
  journal={Cognitive Science},
  volume={46},
  number={5},
  pages={e13143},
  year={2022},
  publisher={Wiley Online Library}
}

@article{yang2023dynamic,
  title={A dynamic computational model of gaze and choice in multi-attribute decisions.},
  author={Yang, Xiaozhi and Krajbich, Ian},
  journal={Psychological Review},
  volume={130},
  number={1},
  pages={52},
  year={2023},
  publisher={American Psychological Association}
}

@article{yang2024attention,
  title={Attention to brand labels affects, and is affected by, evaluations of product attractiveness},
  author={Yang, Xiaozhi and Retzler, Chris and Krajbich, Ian and Ratcliff, Roger and Philiastides, Marios G},
  journal={Frontiers in Behavioral Economics},
  volume={2},
  pages={1274815},
  year={2024},
  publisher={Frontiers Media SA}
}

@article{ting2024unraveling,
  title={Unraveling information processes of decision-making with eye-tracking data},
  author={Ting, Chih-Chung and Gluth, Sebastian},
  journal={Frontiers in Behavioral Economics},
  volume={3},
  pages={1384713},
  year={2024},
  publisher={Frontiers Media SA}
}

@article{zilker2024attentional,
  title={Attentional dynamics of evidence accumulation explain why more numerate people make better decisions under risk},
  author={Zilker, Veronika},
  journal={Scientific Reports},
  volume={14},
  number={1},
  pages={18788},
  year={2024},
  publisher={Nature Publishing Group UK London}
}

@article{thomas2019gaze,
  title={Gaze bias differences capture individual choice behaviour},
  author={Thomas, Armin W and Molter, Felix and Krajbich, Ian and Heekeren, Hauke R and Mohr, Peter NC},
  journal={Nature human behaviour},
  volume={3},
  number={6},
  pages={625--635},
  year={2019},
  publisher={Nature Publishing Group UK London}
}

@misc{lombardi_hare_2021, title={Piecewise constant averaging methods allow for fast and accurate hierarchical Bayesian estimation of drift diffusion models with time-varying evidence accumulation rates}, url={osf.io/preprints/psyarxiv/5azyx}, DOI={10.31234/osf.io/5azyx}, publisher={PsyArXiv}, author={Lombardi, Gaia and Hare, Todd}, year={2021}, month={Oct}}

@book{olver2010nist,
  title={NIST handbook of mathematical functions},
  author={Olver, Frank W and Lozier, Daniel W and Boisvert, Ronald F and Clark, Charles W},
  year={2010},
  publisher={Cambridge university press}
}

@book{lehmann2006theory,
  title={Theory of point estimation},
  author={Lehmann, Erich L and Casella, George},
  year={2006},
  publisher={Springer Science \& Business Media}
}

@book{karatzas1991brownian,
  title={Brownian motion and stochastic calculus},
  author={Karatzas, Ioannis and Shreve, Steven},
  volume={113},
  year={1991},
  publisher={Springer Science \& Business Media}
}

@book{weinan2021applied,
  title={Applied stochastic analysis},
  author={Weinan, E and Li, Tiejun and Vanden-Eijnden, Eric},
  volume={199},
  year={2021},
  publisher={American Mathematical Soc.}
}

@book{oksendal2003stochastic,
  title={Stochastic differential equations},
  author={{\O}ksendal, Bernt and {\O}ksendal, Bernt},
  year={2003},
  publisher={Springer}
}

@article{molter2019glambox,
  title={GLAMbox: A Python toolbox for investigating the association between gaze allocation and decision behaviour},
  author={Molter, Felix and Thomas, Armin W and Heekeren, Hauke R and Mohr, Peter NC},
  journal={PloS one},
  volume={14},
  number={12},
  pages={e0226428},
  year={2019},
  publisher={Public Library of Science San Francisco, CA USA}
}

@book{cox2017theory,
  title={The theory of stochastic processes},
  author={Cox, David Roxbee},
  year={2017},
  publisher={Routledge}
}

@article{navarro2009fast,
  title={Fast and accurate calculations for first-passage times in Wiener diffusion models},
  author={Navarro, Daniel J and Fuss, Ian G},
  journal={Journal of mathematical psychology},
  volume={53},
  number={4},
  pages={222--230},
  year={2009},
  publisher={Elsevier}
}

@article{gondan2014even,
  title={Even faster and even more accurate first-passage time densities and distributions for the Wiener diffusion model},
  author={Gondan, Matthias and Blurton, Steven P and Kesselmeier, Miriam},
  journal={Journal of Mathematical Psychology},
  volume={60},
  pages={20--22},
  year={2014},
  publisher={Elsevier}
}

@article{blurton2012fast,
  title={Fast and accurate calculations for cumulative first-passage time distributions in Wiener diffusion models},
  author={Blurton, Steven P and Kesselmeier, Miriam and Gondan, Matthias},
  journal={Journal of Mathematical Psychology},
  volume={56},
  number={6},
  pages={470--475},
  year={2012},
  publisher={Elsevier}
}

@article{ratcliff1978,
  title={A theory of memory retrieval.},
  author={Ratcliff, Roger},
  journal={Psychological review},
  volume={85},
  number={2},
  pages={59},
  year={1978},
  doi={10.1037/0033-295x.85.2.59},
  publisher={American Psychological Association}
}

@article{doi2020,
  title={The caudate nucleus contributes causally to decisions that balance reward and uncertain visual information},
  author={Doi, Takahiro and Fan, Yunshu and Gold, Joshua I and Ding, Long},
  journal={ELife},
  volume={9},
  pages={e56694},
  year={2020},
  doi={10.7554/eLife.56694},
  publisher={eLife Sciences Publications Limited}
}

@article{yartsev2018,
  title={Causal contribution and dynamical encoding in the striatum during evidence accumulation},
  author={Yartsev, Michael M and Hanks, Timothy D and Yoon, Alice Misun and Brody, Carlos D},
  journal={Elife},
  volume={7},
  pages={e34929},
  year={2018},
  doi={10.7554/elife.34929},
  publisher={eLife Sciences Publications Limited}
}

@article{frank2015,
  title={fMRI and EEG predictors of dynamic decision parameters during human reinforcement learning},
  author={Frank, Michael J and Gagne, Chris and Nyhus, Erika and Masters, Sean and Wiecki, Thomas V and Cavanagh, James F and Badre, David},
  journal={Journal of Neuroscience},
  volume={35},
  number={2},
  pages={485--494},
  year={2015},
  doi={10.1523/JNEUROSCI.2036-14.2015},
  publisher={Soc Neuroscience}
}

@article{diederich2003simple,
  title={Simple matrix methods for analyzing diffusion models of choice probability, choice response time, and simple response time},
  author={Diederich, Adele and Busemeyer, Jerome R},
  journal={Journal of Mathematical Psychology},
  volume={47},
  number={3},
  pages={304--322},
  year={2003},
  publisher={Elsevier}
}

@article{voss2008fast,
  title={A fast numerical algorithm for the estimation of diffusion model parameters},
  author={Voss, Andreas and Voss, Jochen},
  journal={Journal of Mathematical Psychology},
  volume={52},
  number={1},
  pages={1--9},
  year={2008},
  publisher={Elsevier}
}

@article{voss2007fast,
  title={Fast-dm: A free program for efficient diffusion model analysis},
  author={Voss, Andreas and Voss, Jochen},
  journal={Behavior research methods},
  volume={39},
  number={4},
  pages={767--775},
  year={2007},
  publisher={Springer}
}

@article{smith2022modeling,
  title={Modeling evidence accumulation decision processes using integral equations: Urgency-gating and collapsing boundaries.},
  author={Smith, Philip L and Ratcliff, Roger},
  journal={Psychological review},
  volume={129},
  number={2},
  pages={235},
  year={2022},
  publisher={American Psychological Association}
}

@article{smith2000stochastic,
  title={Stochastic dynamic models of response time and accuracy: A foundational primer},
  author={Smith, Philip L},
  journal={Journal of mathematical psychology},
  volume={44},
  number={3},
  pages={408--463},
  year={2000},
  publisher={Elsevier}
}

@article{paninski2008integral,
  title={Integral equation methods for computing likelihoods and their derivatives in the stochastic integrate-and-fire model},
  author={Paninski, Liam and Haith, Adrian and Szirtes, Gabor},
  journal={Journal of Computational Neuroscience},
  volume={24},
  pages={69--79},
  year={2008},
  publisher={Springer}
}

@article{peskir2002integral,
  title={On integral equations arising in the first-passage problem for Brownian motion},
  author={Peskir, Goran},
  journal={Journal of Integral Equations and Applications},
  volume={14},
  number={4},
  pages={397--423},
  year={2002},
  publisher={Rocky Mountain Mathematics Consortium}
}

@article{liu2025efficient,
  title={Efficient Inference in First Passage Time Models},
  author={Liu, Sicheng and Fengler, Alexander and Frank, Michael J and Harrison, Matthew T},
  journal={arXiv preprint arXiv:2503.18381},
  year={2025}
}

@article{saumard2014log,
  title={Log-concavity and strong log-concavity: a review},
  author={Saumard, Adrien and Wellner, Jon A},
  journal={Statistics surveys},
  volume={8},
  pages={45},
  year={2014}
}

@article{thomas2021uncovering,
  title={Uncovering the computational mechanisms underlying many-alternative choice},
  author={Thomas, Armin W and Molter, Felix and Krajbich, Ian},
  journal={Elife},
  volume={10},
  pages={e57012},
  year={2021},
  publisher={eLife Sciences Publications, Ltd}
}

@article{trueblood2026,
  title = {Attentional dynamics explain the elusive nature of context effects.},
  ISSN = {0033-295X},
  url = {http://dx.doi.org/10.1037/rev0000606},
  DOI = {10.1037/rev0000606},
  journal = {Psychological Review},
  publisher = {American Psychological Association (APA)},
  author = {Trueblood,  Jennifer S. and Liu,  Yanjun and Murrow,  Matthew and Hayes,  William M. and Holmes,  William R.},
  year = {2026},
  month = Jan 
}

\newpage
\appendix

\section{Proof of Theorem \ref{thm: 2.1}}\label{sec: proof_prelim}
Here we present a self-contained derivation of both \eqref{eqn: thm211} and \eqref{eqn: thm212} using Kolmogorov forward equations. The connections of diffusion processes and partial differential equations (PDEs) are well-established results in stochastic analysis and can be found in standard references, such as \cite{karatzas1991brownian, oksendal2003stochastic, weinan2021applied}. This approach has the advantage of being easily generalizable to more complex diffusion processes and two-sided boundaries.

\begin{proof}[\proofname \text{ of Theorem \ref{thm: 2.1}.1}]
Without loss of generality we assume $b>x_0$ (i.e. $x=b$ is the upper boundary). Denote the NPD under $\mathbb{P}^{x_0}$ by $u(x, t)$ where $x<b$. It satisfies the following initial-boundary value problem of the Kolmogorov forward equation\footnote{Also known as Fokker-Planck equation in physics.} (see the Chapter 8 of \cite{weinan2021applied}):
\begin{equation}
\begin{alignedat}{2}\label{eqn: fpe}
    u_t &= -\mu u_x + \frac{1}{2}\sigma^2 u_{xx} &\quad & \text{for }t > 0, \; x \in (-\infty, b) \\
    u(x, 0) &= \delta(x - x_0) &\quad &\text{for } x \in (-\infty, b) \\
    u(b, t) &= 0 &\quad &\text{for } t > 0
\end{alignedat}
\end{equation}
where $\delta(\cdot)$ is the Dirac delta measure at $0$. To solve \eqref{eqn: fpe}, let 
\begin{equation*}
u(x, t)=e^{\frac{\mu}{\sigma^2} (x-x_0)-\frac{\mu^2}{2\sigma^2}t}v(x-x_0, t),
\end{equation*} and $v$ is the Green's function that satisfies
\begin{equation*}
\begin{alignedat}{2}
v_t&=\frac{1}{2}\sigma^2 v_{yy} &\quad & \text{for }t > 0, \; y \in (-\infty, b-x_0) \\
v(y, 0)&=\delta(y)&\quad &\text{for } y \in (-\infty, b-x_0) \\
v(b-x_0, t)&=0&\quad &\text{for } t > 0
\end{alignedat}
\end{equation*}
By method of images, we get 
\begin{equation*}
v(y,t)=\frac{1}{\sqrt{2\pi\sigma^2 t}} \Big(e^{-\frac{y^2}{2\sigma^2 t}}-e^{-\frac{(y-2(b-x_0))^2}{2\sigma^2 t}}\Big)
\end{equation*}
Correspondingly, 
\begin{equation}\label{eqn: 221result}
u(x,t)=\frac{1}{\sqrt{2\pi\sigma^2 t}}e^{\frac{\mu}{\sigma^2} (x-x_0)-\frac{\mu^2}{2\sigma^2}t} \Big(e^{-\frac{(x-x_0)^2}{2\sigma^2 t}}-e^{-\frac{(x+x_0-2b)^2}{2\sigma^2 t}}\Big)
\end{equation}
for $x<b$.
\end{proof}
\begin{proof}[\proofname \text{ of Theorem \ref{thm: 2.1}.2}] Here we prove the case where $b>x_0$ and $\mu\ge0$. Denote the density of $\tau_b$ under $\mathbb{P}^{x_0}$ as $f_{\tau_b}(t)$, by definition
\begin{equation*}
f_{\tau_b}(t)=-\frac{\diff}{\diff t} \mathbb{P}^{x_0}(\tau_b>t)
\end{equation*}
$\mathbb{P}^{x_0}(\tau_b>t)=\int_{-\infty}^b u(x,t)\diff x$ represents the net probability flow rate through $x=b$ (see the Chapter 8 of \cite{weinan2021applied}). We can further use the PDE \eqref{eqn: fpe} to continue to write
\begin{equation*}
\begin{aligned}
f_{\tau_b}(t)&=-\frac{\diff}{\diff t} \int_{-\infty}^b u(x,t)\diff x\\
&=- \int_{-\infty}^b u_t(x,t)\diff x\\
&=-\int_{-\infty}^b \Big(-\mu u_x(x,t) + \frac{1}{2}\sigma^2 u_{xx}(x,t)\Big) \diff x\\
&=\mu\int_{-\infty}^b  u_x(x,t)\diff x - \frac{1}{2}\sigma^2 \int_{-\infty}^b u_{xx}(x,t) \diff x\\
&=\mu(u(b, t)-u(-\infty, t)) - \frac{1}{2}\sigma^2 (u_x(b, t)-u_x(-\infty, t))\\
&=- \frac{1}{2}\sigma^2 u_x(b, t)\\
\end{aligned}
\end{equation*}
$u$ is given by \eqref{eqn: 221result}, from which we can compute
\begin{equation*}
\begin{aligned}
f_{\tau_b}(t)
% =&-\frac{1}{2}\sigma^2\left\{\frac{1}{\sqrt{2\pi\sigma^2t}}\Big[\frac{\mu}{\sigma^2} e^{\frac{\mu}{\sigma^2} (x-x_0)-\frac{\mu^2}{2\sigma^2}t}\Big(e^{-\frac{(x-x_0)^2}{2 \sigma^2t}}-e^{-\frac{(x+x_0-2b)^2}{2 \sigma^2 t}}\Big) \right.\\
% &\left.+e^{\frac{\mu}{\sigma^2} (x-x_0)-\frac{\mu^2}{2\sigma^2}t}\Big(-e^{-\frac{(x-x_0)^2}{2 \sigma^2t}}\frac{x-x_0}{\sigma^2 t}+e^{-\frac{(x+x_0-2b)^2}{2 \sigma^2 t}}\frac{x+x_0-2b}{\sigma^2 t}\Big) \Big]\right\}_{x=b}\\
% =&-\frac{\sigma^2}{2\sqrt{2\pi\sigma^2t}}\Big[\frac{\mu}{\sigma^2} e^{\frac{\mu}{\sigma^2}(b-x_0)-\frac{ \mu^2}{2 \sigma^2}t}\Big(e^{-\frac{(b-x_0)^2}{2 \sigma^2t}}-e^{-\frac{(b+x_0-2b)^2}{2 \sigma^2 t}}\Big) \\
% &+e^{\frac{\mu }{\sigma^2}(b-x_0)-\frac{ \mu^2}{2 \sigma^2}t}\Big(-e^{-\frac{(b-x_0)^2}{2 \sigma^2t}}\frac{b-x_0}{\sigma^2 t}+e^{-\frac{(b+x_0-2b)^2}{2 \sigma^2 t}}\frac{b+x_0-2b}{\sigma^2 t}\Big)\Big]\\
% =&\frac{\sigma^2}{\sqrt{2\pi\sigma^2t}}e^{\frac{\mu }{\sigma^2}(b-x_0)-\frac{ \mu^2}{2 \sigma^2}t}e^{-\frac{(b-x_0)^2}{2 \sigma^2 t}}\frac{b-x_0}{\sigma^2 t}\\
=&\frac{b-x_0}{\sqrt{2 \pi \sigma^2 t^3}} e^{-\frac{\left(b-x_0-\mu t\right)^2}{2 \sigma^2 t}} 
\end{aligned}
\end{equation*}

\end{proof}

\section{Log-concavity of the Likelihood \eqref{eqn: addmllhd}}\label{sec: logconcave}
In this section we prove that the likelihood function \eqref{eqn: addmllhd} is a strict log-concave function of $\mu$. It suffices to show that $\frac{\partial^2}{\partial\mu^2}\log f_{\text{aDDM}}(\tau;\mu)<0$.

 When $\tau \le T$, we have $\frac{\partial^2}{\partial\mu^2}\log f_\text{aDDM}(\tau;\mu)=-\tau <0$; When $\tau > T$, we write
\begin{equation*}
\begin{aligned}
f_\text{aDDM}(\tau;\mu)=&\int_{-\infty}^b \frac{1}{\sqrt{2\pi T}}e^{\mu x-\frac{\mu^2}{2}T} \Big(e^{-\frac{x^2}{2T}}-e^{-\frac{(x-2b)^2}{2T}}\Big)\frac{b-x}{\sqrt{2\pi (\tau-T)^3}} e^{-\frac{(b-x)^2}{2(\tau-T)}} \diff x\\
\triangleq& \frac{e^{-\frac{\mu^2}{2}T}}{2\pi\sqrt{ T(\tau-T)^3}} \int_{-\infty}^b e^{\mu x} g(x) \diff x
\end{aligned}
\end{equation*}
where
\begin{equation*}
g(x)\triangleq (b-x)\Big(e^{-\frac{x^2}{2T}}-e^{-\frac{(x-2b)^2}{2T}}\Big) e^{-\frac{(b-x)^2}{2(\tau-T)}}
\end{equation*}
is independent of $\mu$. When $\mu$ is in any compact sub-interval in $\mathbb{R}$, $|x e^{\mu x} g(x)|$ and $|x^2 e^{\mu x} g(x)|$ can be easily dominated by $L^1$ functions, hence by Lebesgue's dominated convergence theorem we have
\begin{equation*}
\begin{aligned}
\frac{\partial }{\partial \mu} \int_{-\infty}^b e^{\mu x} g(x)\diff x&= \int_{-\infty}^b x e^{\mu x} g(x)\diff x\\
\frac{\partial^2 }{\partial \mu^2} \int_{-\infty}^b e^{\mu x} g(x)\diff x&= \int_{-\infty}^b x^2 e^{\mu x} g(x)\diff x\\
\end{aligned}
\end{equation*}
and consequently,
\begin{equation}\label{eqn: 2nd_deriv}
\begin{aligned}
% \log f_\text{aDDM}(\tau;\mu)&=-\frac{T}{2}\mu^2+\log\int_{-\infty}^b e^{\mu x} g(x)\diff x\\
% \frac{\partial}{\partial \mu}\log f_{\text{aDDM}}(\tau;\mu)&=-T\mu +\frac{\int_{-\infty}^b xe^{\mu x} g(x)\diff x}{\int_{-\infty}^b e^{\mu x} g(x)\diff x}\\
\frac{\partial^2}{\partial \mu^2}\log f_{\text{aDDM}}(\tau;\mu)&=-T +\frac{\int_{-\infty}^b x^2 e^{\mu x} g(x)\diff x}{\int_{-\infty}^b e^{\mu x} g(x)\diff x }-\left(\frac{\int_{-\infty}^b xe^{\mu x} g(x)\diff x}{\int_{-\infty}^b e^{\mu x} g(x)\diff x }\right)^2
\end{aligned}
\end{equation}
Define the exponential tilting of $g$ as
\begin{equation*}
h_\mu(x)=\frac{e^{\mu x} g(x)}{\int_{-\infty}^b e^{\mu x} g(x)\diff x}
\end{equation*}
$h_\mu(x)$ is a probability density function on $(-\infty, b)$. Let $Y_\mu$ be a random variable with density $h_\mu(x)$, then we can interpret Eq.\eqref{eqn: 2nd_deriv} as
\begin{equation}\label{eqn: 2nd_deriv_var}
\frac{\partial^2}{\partial\mu^2} \log f_\text{aDDM}(\tau;\mu)=-T+\var[Y_\mu]
\end{equation}
Let $V(x)=-\log h_\mu(x)$ so that $h_{\mu}(x)=e^{-V(x)}$, we can compute that
\begin{equation*}
V''(x)=\frac{1}{T}+\frac{1}{\tau-T}+\frac{1}{(x-b)^2}+\Big(\frac{2b}{T}\Big)^2\frac{e^{\frac{2b}{T} (x-b)}}{\big(1- e^{\frac{2b}{T}(x-b)}\big)^2}> \frac{1}{T}+\frac{1}{\tau-T}
\end{equation*}
i.e., $V$ is $\frac{1}{T}+\frac{1}{\tau-T}$-strongly convex. By Poincar\'{e}'s inequality for strictly log-concave measures (see, for example Proposition 10.1 of \cite{saumard2014log}) we have
\begin{equation*}
\var[Y_\mu]\le \frac{1}{\frac{1}{T}+\frac{1}{\tau-T}}=\frac{\tau-T}{\tau} T< T
\end{equation*}
Hence from Eq.\eqref{eqn: 2nd_deriv_var} we have $\frac{\partial^2}{\partial\mu ^2}\log f_{\text{aDDM}} (\tau; \mu) <0$ for $\tau > T$.\qed

\section{Regularity Conditions for the Consistency of $\widehat{\mu}_n^{\text{ML}}$}\label{sec: regularities}
When $\tau_1, \cdots, \tau_n$ are independent and identically distributed as $\tau$, we can show the consistency of $\widehat{\mu}_n^{\text{ML}}=\widehat{\mu}_n^{\text{ML}}(\tau_1, \cdots, \tau_n)$.

First, we write the likelihood equation
\begin{equation}\label{eqn: llhd_eqn}
\frac{\partial}{\partial\mu} \ell_n^{\text{aDDM}}(\mu\mid \tau_1, \cdots, \tau_n)=0
\end{equation}
where $\ell_n^{\text{aDDM}}(\mu\mid \tau_1, \cdots, \tau_n)\triangleq-\sum_{i=1}^n \log f_{\text{aDDM}}(\tau_i;\mu)$ is the negative log-likelihood function. We prove that Eq.\eqref{eqn: llhd_eqn} has a root $\widehat{\mu}_n=\widehat{\mu}_n(\tau_1, \cdots, \tau_n)$ that converges to the true value $\mu_0$ in probability by verifying the following regularity conditions (see theorem 6.3.7 of \cite{lehmann2006theory}):
\begin{itemize}
    \item Identifiability: We aim to show that  $f_{\text{aDDM}}(\tau; \mu_1)=f_{\text{aDDM}}(\tau;\mu_2)$ for all $\tau>0$ implies $\mu_1=\mu_2$. Considering when $\tau\le T$, $f_{\text{aDDM}}(\tau;\mu_1)=f_{\text{aDDM}}(\tau;\mu_2)$ immediately leads to
\begin{equation*}
(b-\mu_1\tau)^2=(b-\mu_2\tau)^2
\end{equation*}
holds for any $0<\tau \le T$, so either $\mu_1=\mu_2$ or $(\mu_1+\mu_2)\tau=2b$ for $0<\tau \le T$. The latter case implies $\mu_1+\mu_2=b=0$, which contradicts the fact that $b$ is positive, so $\mu_1$ must equal $\mu_2$. When $\tau>T$, $\mu_1=\mu_2$ automatically gives
$f_{\mathrm{aDDM}}(\tau;\mu_1)=f_{\mathrm{aDDM}}(\tau;\mu_2)$. Hence, we conclude that $\mu_1=\mu_2$.

\item Common support: We claim that for any $\mu\in \mathbb{R}$, $f_{\text{aDDM}}(\tau;\mu)$ has common support $\tau\in(0,\infty)$. When $0<\tau\le T$ this is obvious; when $\tau>T$, lemma \ref{lemma: 2.4} implies $\varphi\Big((T\mu+b) \sqrt{\frac{\tau-T}{2T\tau}}\Big)>\varphi\Big((T\mu-b) \sqrt{\frac{\tau-T}{2T\tau}}\Big)$ where $\varphi(x)\triangleq xe^{x^2}\operatorname{Erfc}(x)$. This inequality establishes the positivity of $f_{\text{aDDM}}(\tau;\mu)$.

\item The true $\mu$ lies in the interior of the parameter space $\mathbb{R}$ as $\mathbb{R}$ is an open set.

\item For a.e. $\tau$, $f_{\text{aDDM}}(\tau;\mu)$ is differentiable w.r.t. $\mu$.

\end{itemize}
so $\widehat{\mu}_n$ converges in probability to the true $\mu$ as $n\rightarrow \infty$. 

In \ref{sec: logconcave} we have proven that $\log f_{\text{aDDM}}(\tau;\mu)$ is a strictly concave function of $\mu$, so $\widehat{\mu}_n$ is in fact the maximum likelihood estimator $\widehat{\mu}_n^{\text{ML}}$. So  $\widehat{\mu}_n^{\text{ML}}$ is consistent.

\section{Supplemental Derivations}
\subsection{Derivation of Eq.\eqref{eqn: addmllhd}}\label{sec: deriv1}
When $\tau>T$, we have
\begin{equation*}
\begin{aligned}
&f_{\text{aDDM}}(\tau;\mu)\\
=&\int_{-\infty}^b \tfrac{1}{\sqrt{2\pi T}}e^{\mu x-\tfrac{\mu^2}{2}T} \Big(e^{-\frac{x^2}{2T}}-e^{-\frac{(x-2b)^2}{2T}}\Big)\tfrac{b-x}{\sqrt{2\pi (\tau-T)^3}} e^{-\frac{(b-x)^2}{2(\tau-T)}} \diff x\\
\triangleq& I_1-I_2\\
\end{aligned}    
\end{equation*}
where
\begin{equation*}
\begin{aligned}
I_1&\triangleq\int_{-\infty}^b \tfrac{1}{\sqrt{2\pi T}}e^{\mu x-\frac{\mu^2}{2}T} e^{-\frac{x^2}{2T}}\tfrac{b-x}{\sqrt{2\pi (\tau-T)^3}} e^{-\frac{(b-x)^2}{2(\tau-T)}} \diff x\\
I_2&\triangleq\int_{-\infty}^b \tfrac{1}{\sqrt{2\pi T}}e^{\mu x-\frac{\mu^2}{2}T} e^{-\frac{(x-2b)^2}{2T}}\tfrac{b-x}{\sqrt{2\pi (\tau-T)^3}} e^{-\frac{(b-x)^2}{2(\tau-T)}} \diff x
\end{aligned}
\end{equation*}
We can compute
\begin{alignat*}{1}
I_1
=&\int_{-\infty}^b \tfrac{1}{\sqrt{2\pi T}}\tfrac{b-x}{\sqrt{2\pi (\tau-T)^3}} e^{\mu x-\frac{\mu^2}{2}T-\frac{x^2}{2T}-\frac{(b-x)^2}{2(\tau-T)}} \diff x\\
&\text{(Complete the square)}\\
=&\int_{-\infty}^b \tfrac{1}{\sqrt{2\pi T}}\tfrac{b-x}{\sqrt{2\pi (\tau-T)^3}} e^{-\frac{(T\mu-b)^2}{2\tau}} e^{-\frac{\tau}{2T (\tau-T)}(x-T\mu-\frac{Tb}{\tau}+\frac{T^2\mu}{\tau})^2}\diff x
\\
&\text{(Let $M_1=\tfrac{\tau-T}{\tau}(T\mu-b), N =\tfrac{\tau-T}{\tau}$)}\\
=&\tfrac{e^{-\frac{(T \mu-b)^2}{2 \tau}}}{2 \pi \sqrt{T(\tau-T)^{3}}} \int_{-\infty}^b (b-x) e^{-\frac{(x-M_1-b)^2}{2T N }}\diff x\\
% =&\int_{-\infty}^b \tfrac{1}{\sqrt{2\pi T}}\tfrac{b-x}{\sqrt{2\pi (\tau-T)^3}} e^{-\frac{(T\mu-b)^2}{2\tau}} e^{-\frac{\tau}{2T (\tau-T)}(x-M_1-b)^2}\diff x\\
% =&\tfrac{e^{-\frac{(T \mu-b)^2}{2 \tau}}}{2 \pi \sqrt{T(\tau-T)^{3}}} \int_0^{\infty} y e^{-\frac{\tau}{2 T(\tau-T)}(y+M_1)^2} \diff y (y=b-x)\\
&\text{(Let $y=-x+M_1+b$)}\\
=&\tfrac{e^{-\frac{(T \mu-b)^2}{2 \tau}}}{2 \pi \sqrt{T(\tau-T)^{3}}} \int_{M_1}^{\infty} (y-M_1) e^{-\frac{y^2}{2 TN }} \diff y\\
=&\tfrac{e^{-\frac{(T \mu-b)^2}{2 \tau}}}{2 \pi \sqrt{T(\tau-T)^{3}}} \left(\int_{M_1}^{\infty} ye^{-\frac{y^2}{2 TN }} \diff y-M_1
\int_{M_1}^{\infty} e^{-\frac{y^2}{2 TN }} \diff y\right)\\
=&\tfrac{e^{-\frac{(T \mu-b)^2}{2 \tau}}}{2 \pi \sqrt{T(\tau-T)^{3}}} \Big(TN  e^{-\frac{M_1^2}{2TN }}-M_1 \sqrt{\tfrac{\pi T N }{2}}\operatorname{Erfc}\Big(\tfrac{M_1}{\sqrt{2TN }}\Big)\Big)\\
% =&\tfrac{e^{-\frac{(T \mu-b)^2}{2 \tau}}}{2 \pi \sqrt{T(\tau-T)}} \tfrac{T}{\tau}\left[e^{-\frac{(T\mu-b)^2(\tau-T)}{2T\tau}}-(T \mu-b) \sqrt{\tfrac{\pi}{2} \tfrac{\tau-T}{T\tau}}\operatorname{Erfc}\Big((T\mu-b)\sqrt{\tfrac{\tau-T}{2T\tau}}\Big)\right]\\
&\text{(Substitute $M_1, N$, and simplify)}\\
=&\tfrac{T e^{-\frac{(T \mu-b)^2}{2 T}}}{2 \pi \tau\sqrt{T(\tau-T)}}-
\tfrac{e^{-\frac{(T \mu-b)^2}{2 \tau}}}{2 \sqrt{2\pi\tau^3} } (T \mu-b) \operatorname{Erfc}\Big((T\mu-b)\sqrt{\tfrac{\tau-T}{2T\tau}}\Big)
\end{alignat*}
and similarly
\begin{alignat*}{1}
I_2=&\int_{-\infty}^b \tfrac{1}{\sqrt{2\pi T}}\tfrac{b-x}{\sqrt{2\pi (\tau-T)^3}} e^{\mu x-\frac{\mu^2}{2}T-\frac{(x-2b)^2}{2T}-\frac{(b-x)^2}{2(\tau-T)}} \diff x\\
&\text{(Complete the square)}\\
=&\int_{-\infty}^b \tfrac{1}{\sqrt{2\pi T}}\tfrac{b-x}{\sqrt{2\pi (\tau-T)^3}} e^{-\frac{(T\mu+b)^2}{2\tau}+2\mu b}  e^{-\frac{\tau}{2T(\tau-T)}(x-T\mu+\frac{Tb}{\tau}+\frac{T^2\mu}{\tau}-2b)^2} \diff x\\
&\text{(Let $M_2=\tfrac{\tau-T}{\tau}(T\mu+b), N=\tfrac{\tau-T}{\tau}$)}\\
=&\tfrac{e^{-\frac{(T\mu+b)^2}{2\tau}+2\mu b}}{2\pi\sqrt{T (\tau-T)^3}}\int_{-\infty}^b (b-x) e^{-\frac{(x-M_2-b)^2}{2TN}}  \diff x\\
&\text{(Let $y=-x+M_2+b$)}\\
=&\tfrac{e^{-\frac{(T\mu+b)^2}{2\tau}+2\mu b}}{2\pi\sqrt{T (\tau-T)^3}}\int_{M_2}^\infty (y-M_2) e^{-\frac{y^2}{2TN}}  \diff y
\\
=&\tfrac{e^{-\frac{(T\mu+b)^2}{2\tau}+2\mu b}}{2\pi\sqrt{T (\tau-T)^3}}\left(\int_{M_2}^\infty y e^{-\frac{y^2}{2TN}}  \diff y-M_2\int_{M_2}^\infty e^{-\frac{y^2}{2TN}}  \diff y\right)\\
=&\tfrac{e^{-\frac{(T\mu+b)^2}{2\tau}+2\mu b}}{2\pi\sqrt{T (\tau-T)^3}}\Big(TN e^{-\frac{M_2^2}{2TN}}-M_2 \sqrt{\tfrac{\pi T N}{2}}\operatorname{Erfc}\Big(\tfrac{M_2}{\sqrt{2TN}}\Big)\Big)\\
% =&\tfrac{e^{-\frac{(T \mu+b)^2}{2 \tau}+2 \mu b}}{2 \pi \sqrt{T(\tau-T)^3}}\left[\tfrac{T(\tau-T)}{\tau} e^{-\frac{(T\mu+b)^2(\tau-T)}{2T\tau}}-\tfrac{\tau-T}{\tau}(T\mu+b)\sqrt{\tfrac{\pi T(\tau-T)}{2 \tau}} \operatorname{Erfc}\Big((T\mu +b) \sqrt{\tfrac{\tau-T}{2 T \tau}}\Big)\right]\\
% =&\tfrac{e^{-\frac{(T \mu+b)^2}{2 \tau}+2 \mu b}}{2 \pi \sqrt{T(\tau-T)}}\tfrac{T}{\tau}\left[ e^{-\frac{(T\mu+b)^2(\tau-T)}{2T\tau}}+(T\mu+b)  \sqrt{\tfrac{\pi}{2}\tfrac{\tau-T}{ T\tau}} \operatorname{Erfc}\Big((T\mu +b) \sqrt{\tfrac{\tau-T}{2 T \tau}}\Big)\right]\\
&\text{(Substitute $M_2, N$, and simplify)}\\
=&\tfrac{T e^{-\frac{(T \mu-b)^2}{2 T}}}{2 \pi \tau\sqrt{T(\tau-T)}}-
\tfrac{e^{-\frac{(T \mu+b)^2}{2 \tau}+2 \mu b}}{2 \sqrt{2\pi \tau^3} }(T\mu+b) \operatorname{Erfc}\Big((T\mu +b) \sqrt{\tfrac{\tau-T}{2 T \tau}}\Big)
\end{alignat*}
Subtract $I_2$ from $I_1$ yields the desired result in \eqref{eqn: addmllhd}.\qed

\subsection{Derivation of Eq.\eqref{eqn: 18}}\label{sec: deriv2}
\begin{align*}
\mathbb{E}[\tau\mathbbm{1}_{\tau\le T}]=&\int_0^T \tfrac{b}{\sqrt{2\pi t}}e^{-\frac{(\mu t-b)^2}{2t}} \diff t\\
=&\tfrac{b}{\sqrt{2\pi} \mu } \int_0^T \left(\tfrac{\mu t+b}{2t^{3/2}}+\tfrac{\mu t-b}{2t^{3/2}}\right)e^{-\frac{(\mu t-b)^2}{2t}} \diff t\\
=&\tfrac{b}{\sqrt{2\pi} \mu } \bigg(\int_0^T \tfrac{\mu t+b}{2t^{3/2}}e^{-\frac{(\mu t-b)^2}{2t}} \diff t+e^{2b\mu}\int_0^T \tfrac{\mu t-b}{2t^{3/2}}e^{-\frac{(\mu t+b)^2}{2t}} \diff t\bigg)\\
&\text{(Let $y=\tfrac{\mu t-b}{\sqrt{t}}, z=\tfrac{\mu t+b}{\sqrt{t}}$, then $\diff y=\tfrac{\mu t +b}{2 t^{3/2}}\diff t, \diff z=\tfrac{\mu t- b}{2 t^{3/2}}\diff t$)}\\
=&\tfrac{b}{\sqrt{2\pi} \mu } \bigg(\int_{-\infty}^{\frac{T\mu-b}{\sqrt{T}}} e^{-\frac{y^2}{2}} \diff y+e^{2b\mu}\int_\infty^{\frac{T\mu +b}{\sqrt{T}}} e^{-\frac{z^2}{2}} \diff z\bigg)\\
% =&\tfrac{b}{\sqrt{2\pi} \mu } \bigg(\sqrt{2 \pi} \cdot \tfrac{1}{2} \Big(2-\operatorname{Erfc}\Big(\tfrac{T \mu-b}{\sqrt{2 T}}\Big)\Big)+e^{2b\mu}\sqrt{2 \pi} \cdot \tfrac{1}{2} \operatorname{Erfc}\left(\tfrac{\mu T+b}{\sqrt{2 T}}\right)\bigg)\\
=&\tfrac{b}{\mu}\left(1-\tfrac{1}{2}\operatorname{Erfc}\Big(\tfrac{T\mu -b}{\sqrt{2T}}\Big)-\tfrac{1}{2}e^{2 b \mu} \operatorname{Erfc}\Big(\tfrac{T\mu +b}{\sqrt{2T}}\Big)\right)
\end{align*}
\qed

\end{document}